\documentclass[12pt,a4paper]{article}
\textheight      = \paperheight
\advance\textheight by -2truein
\textwidth       = \paperwidth
\advance\textwidth  by -2truein
\usepackage[T1]{fontenc}

\usepackage{indentfirst,eurosym}
\usepackage{a4wide}
\usepackage{amsmath,amsthm,amsfonts,amscd,amssymb,bbm,mathrsfs,enumerate,url}
\usepackage{color}

\usepackage{graphicx,tikz}
\usepackage[pdfborder={0 0 0}]{hyperref}
\usetikzlibrary{arrows}

\usepackage{ucs}
\usepackage[utf8x]{inputenc}

\def\RR{{\mathbb R}}
\def\CC{{\mathbb C}}
\def\NN{{\mathbb N}}
\def\ZZ{{\mathbb Z}}

\def\a{{\mathcal A}}
\def\b{{\mathcal B}}
\def\h{{\mathcal H}}
\def\d{{\mathcal D}}
\def\diff{{{\rm Diff}^+ (S^1)}}
\def\diffI{{{\rm Diff}^+ (I)}}

\def\Vect{{{\rm Vect}}}

\def\exp{{\rm Exp}}
\def\MM{{\mathbb M}}
\def\mob{{\mathrm{M\ddot{o}b}}}
\def\A{{\mathcal A}}
\def\C{{\mathcal C}}
\def\I{{\mathcal I}}
\def\K{{\mathcal K}}
\def\M{{\mathcal M}}

\def\N{{\mathcal N}}
\def\P{{\mathcal P}}
\def\QQ{{\mathbb Q}}
\def\R{{\mathcal R}}
\def\T{{\mathbb T}}
\def\U{{\mathcal U}}
\def\g{\gamma}

\def\s{\sigma}

\def\Exp{{\mathrm{Exp}\,}}

\def\<{\langle}
\def\>{\rangle}
\def\Ad{{\hbox{\rm Ad\,}}}
\def\Tr{{\hbox{\rm Tr\,}}}
\def\supp{{\rm supp\,}}
\def\dom{{\rm Dom\,}}
\def\1{{\mathbbm 1}}
\def\cyl{\mathcal{E}}
\def\conf{\mathrm{Conf}(\mathcal{E})}
\def\cprime{$'$} 

\newtheorem{theorem}{Theorem}[section]
\newtheorem{definition}[theorem]{Definition}
\newtheorem{corollary}[theorem]{Corollary}
\newtheorem{proposition}[theorem]{Proposition}
\newtheorem{lemma}[theorem]{Lemma}
\theoremstyle{remark}
\newtheorem{remark}[theorem]{Remark}

\begin{document}
\title{Conformal covariance and the split property}
\author{
{\bf Vincenzo Morinelli}\footnote{Supported by the ERC advanced grant 669240 QUEST ``Quantum Algebraic Structures and Models''.},\quad
{\bf Yoh Tanimoto}\footnote{Supported by the JSPS fellowship for research abroad.}\\
Dipartimento di Matematica, Universit\'a di Roma Tor Vergata\\
Via della Ricerca Scientifica, 1, I-00133 Roma, Italy\\
email: {\tt morinell@mat.uniroma2.it},\quad {\tt hoyt@mat.uniroma2.it}
\vspace{0.5cm}\\
{\bf Mih\'aly Weiner}\footnote{Supported in part by the ERC advanced grant 669240 QUEST ``Quantum Algebraic Structures and Models'' and by OTKA grant no. 104206.}\\
Mathematical Institute, Department of Analysis, \\
Budapest University of Technology \& Economics (BME)\\
M\H{u}egyetem rk. 3-9, H-1111 Budapest, Hungary\\
email: {\tt mweiner@math.bme.hu}
}
\date{}
\maketitle

\begin{abstract}
%
We show that for a conformal local net of observables on the circle, the {\it split} property is automatic.
Both {\it full} conformal covariance (i.e.\! diffeomorphism covariance) and the circle-setting play essential roles
in this fact, while by previously constructed examples it was already known that even on the circle,
M\"obius covariance does {\it not} imply the split property.

On the other hand, here we also provide an example of a local conformal net living on the $2$-dimensional Minkowski space,
which --- although being diffeomorphism covariant --- does not have the split property.
\end{abstract}

\section{Introduction}\label{introduction}
More than half a century passed away since the first formulation of an {\it axiomatic} quantum field theory. There are several existing different settings (differing e.g.\! on the chosen spacetime, or whether their fundamental notion is that of a {\it quantum field}
or a {\it local observable}) with many ``additional'' properties that are sometimes included among the defining axioms. For an introduction and overview of the topic we refer to the book of Haag \cite{Haag96}.

Whereas properties like {\it locality} are unquestionably among the basic axioms, some other properties are less motivated and accepted. {\it Haag-duality} has an appealing mathematical elegance, but there seems to be no clear physical motivation for that assumption.  Technicalities, like the {\it separability} of the underlying Hilbert space are sometimes required with no evident physical reason. 

The {\it split property} is the statistical independence of local algebras associated to regions with a positive (spacelike) separation. It might be viewed as a stronger version of locality, and contrary to the previous two examples, it was formulated on direct physical grounds. However, traditionally it is not included among the defining axioms, as in the beginning it was unclear how much one can believe in it. Indeed, many years passed till this stronger version of locality was first established at least for the massive free field by Buchholz \cite{Buchholz74}.
Only after the introduction of  the nuclearity condition (which was originally motivated by the need of a particle interpretation \cite{HS65}) it became more of a routine to verify the split property in various models, when its connection to nuclearity was discovered \cite{BW86}.
Another important step was the general mathematical understanding of split inclusions brought by the work by Doplicher and Longo \cite{DL}.

In the meantime, interest rose in conformal quantum field theories, especially in the low dimensional case; i.e.\! conformal models given on the 2-dimensional Minkowski space and their chiral components that can be naturally extended onto the compactified lightray, the circle. The theory of conformal net of local algebras on $S^1$ is rich in examples 
and it provides an essential ``playground'' to people studying operator algebras as it turned out to have incredibly deep connections to the modular theory of von Neumann algebras as well as to subfactor theory; see e.g.\! \cite{FrG,Wiesbrock93-2,KLM01}. In particular, the modular group associated to a local algebra and the vacuum vector always acts in a certain geometric manner: the so-called {\it Bisognano-Wichmann property} is automatic. In turn, this was used to conclude that several further important structural properties --- e.g. {\it Haag-duality} and {\it Additivity} --- are also automatic in this setting. We refer to the original works \cite{FrG,FJ,BGL93,BDL} for more details on this topic.

The case of the split property seemed to be different --- but there is an important detail to mention here. Initially, when studying chiral conformal nets, only M\"obius covariance was exploited in the so-far cited works. There were several reasons behind this choice.
First, M\"obius symmetry is the spacetime symmetry implemented by a unitary representation for which the vacuum is an invariant vector. This is exactly how things go in higher dimension, but this is {\it not} how diffeomorphism covariance is implemented (no invariant vectors and one is forced to consider {\it projective} representations rather than true ones).
Second, the mentioned connection to modular theory of von Neumann algebras relies on M\"obius covariance only. Thus, the listed structural properties --- with the exception of the split property --- are already automatic even if diffeomorphism covariance is not assumed.      

From the physical point of view, however, diffeomorphism covariance is natural in the low dimensional conformal setting; by an argument of L\"uscher and Mack, it should merely be a consequence of the existence of a stress-energy tensor \cite{FST89}. All important models are diffeomorphism covariant with the exception of some ``pathological'' counter-examples; see \cite{Koester03, CW}. It is worth noting that the example constructed in \cite{CW} by infinite tensor products, has {\it neither} diffeomorphism symmetry {\it nor} the split property. Thus, unlike the mentioned other properties, the split property surely cannot be derived in the M\"obius covariant setting. However, as we shall prove it here, the split property is automatic if diffeomorphism covariance is assumed. Note that together with the result of Longo and Xu in \cite{LX} regarding strong additivity, this shows that a diffeomorphism covariant local net on $S^1$ is {\it completely rational} if and only if its $\mu$-index is finite.

The crucial points of our proof are the following. We consider a conformal net $\A$ on the circle with conformal Hamiltonian $L_0$, and fix two (open, proper) intervals $I_a,I_b\in \I$ with positive distance from each other. Inspired by the complex analytic argument used in \cite{FJ} to prove the conformal cluster theorem, for an element 
$X$ of the $*$-algebra $\a(I_a)\vee_{\mathrm{alg}}\a(I_b)$ generated by $\a(I_a)$ and $\a(I_b)$
with decomposition $X=\sum_{k=1}^n A_k B_k$ (where $n\in\NN,\, A_k\in \a(I_a),\, B_k\in \a(I_b)$), we consider the function 
on the complex unit disc 
\begin{equation*}
z\mapsto \sum_{k=1}^n \langle\Omega, A_k z^{L_0} B_k\Omega\rangle.
\end{equation*}
For every $|z|\leq 1$, this defines a functional $\phi_z$ on $\a(I_a)\vee_{\mathrm{alg}}\a(I_b)$.
For $z=1$ this is simply the vacuum state $\omega$, but for $z=0$ this is the product vacuum state $AB\mapsto \omega(A)\omega(B)$ $(A\in \a(I_a),\, B\in \a(I_b))$. The split property is essentially equivalent to saying that $\phi_0$ is normal 
(actually, here some care is needed: in general one needs the product state to be normal {\it and} faithful. Fortunately, general results on normality and conormality in a M\"obius covariant net \cite{GLW} of the inclusions $\a(I_1)\subset \a(I_2)$ for an $I_1\subset I_2$ imply that $\a(I_a)\vee \a(I_b)$ is a factor; see more details in the preliminaries. It then turns out that the normality of $\phi_0$ is indeed equivalent to the split property).

However, we do not have a direct method to show that $\phi_z$ is normal at $z=0$. On the other hand, we {\it can} treat several points inside the disc. Using the positive energy projective representation $U$ of $\diff$ given with the theory, for example for any  (fixed) $r\in (0,1)$ and $I_c,I_d\in \I$ covering the full circle we find a decomposition $r^{L_0} =C D$ in which $C_r\in \a(I_c)$ and $D_r\in \a(I_d)$.
Choosing the intervals $I_c$ and $I_d$ carefully, $C$
will commute with the $A_k$ operators while $D$ will commute with the $B_k$ operators and hence 
\begin{eqnarray*}
\nonumber
\phi_r(X) &=& \sum_{k=1}^n \langle\Omega, A_k r^{L_0} B_k\Omega\rangle = 
\sum_{k=1}^n \langle\Omega, A_k C D B_k\Omega\rangle \\
&=& \sum_{k=1}^n \langle C^*\Omega, A_k B_kD\Omega\rangle =
\langle C^*\Omega,\, X\,D\Omega\rangle 
\end{eqnarray*}
showing that for our real $r\in (0,1)$, the functional $\phi_r$ is normal as it is given by two vectors. Note that the origin of the decomposition $r^{L_0} =C D$ is the fact that
a rotation can be decomposed as a product of {\it local} diffeomorphisms; something that using M\"obius transformations alone, cannot be achieved (as all nontrivial M\"obius transformations are global). However, even using the full diffeomorphism group, the issue is tricky, since we need a decomposition that can be analytically continued over to some imaginary parameters --- and of course the words ``local'' and ``analytical'' are usually in conflict with each other. Nevertheless, this kind of problem was already treated in 
\cite{weiner}, and the methods there developed were also used in the proof of \cite[Theorem 2.16]{CCHW}, so all we needed here was some adaptation of earlier arguments. 

We then proceed by ``deforming'' our decomposition using the work \cite{Ols} of Olshanskii,
which allows us to access further regions inside the unit disk. In this way we establish normality along a ring encircling the origin, and thus we can use the Cauchy integral formula to conclude normality of $\phi_z$ at $z=0$. 

Note that we have really made use of the fact that the conformal Hamiltonian $L_0$ generates a compact group. Indeed, for a generic complex number $z$, the very expression $z^{L_0}$ is meaningful only because ${\rm Sp}(L_0)$ contains integer values only. However, unlike with chiral nets, in the $2$-dimensional conformal case the theory does not necessarily extends in a natural way to the compactified spacetime. Thus one might wonder whether our result will remain valid or not: is this compactness of the spacetime just some technicality, or is it an essential ingredient of our proof? The answer turns out to be the latter one. 

In fact, we manage to present an example of a diffeomorphism covariant local net on the $2$-dimensional spacetime,
which does not have the split property. More concretely, we consider a local extension $\tilde \a\supset \a$ of
the net $\a=\a_{U(1)}\otimes \a_{U(1)}$ obtained by taking two copies of the $U(1)$-current net (here considered
as ``left'' and ``right'' chiral parts). Irreducible sectors of the $U(1)$-current net are classified by a certain charge $q\in \RR$.
Our construction is such that when considered as a representation of $\a_{U(1)}\otimes \a_{U(1)}$, the net $\a\subset \tilde \a$
decomposes as a direct sum  $\oplus_{q\in \QQ}\, (\s_q\otimes\s_q)$ where $\s_q$ is the representation corresponding to the sector with charge $q$.
This model is naturally diffeomorphism covariant, but as it violates the modular compactness \cite{BDL90},
it cannot have the split property. Note that here ``diffeomorphism covariance'' means only that
we have an action of $\widetilde{\diff}\times \widetilde{\diff}$ which factors through the spacelike $2\pi$-rotation,
but not that of $\diff\times \diff$. This is in complete accordance with our earlier remark on the spectrum of $L_0$.
It is also possible to replace $\QQ$ with $\RR$ to obtain a diffeomorphism covariant net
on a non-separable Hilbert space, and it is immediate to show that it does not have the split property.

This paper is organized as follows.
In Section \ref{preliminaries} we introduce our operator-algebraic setting for conformal field theory and
recall relevant technical results concerning conformal covariance and the split property.
Sections \ref{localdecomp} and \ref{furtherdecomp} provide our technical ingredients, namely
certain decompositions of $z^{L_0}$ into local elements.
In Section \ref{normality} we prove our main result, that the split property follows from diffeomorphism covariance,
by proving the normality of $\phi_0$.
A two-dimensional counterexample is provided in Section \ref{non-split}.
In Section \ref{outlook} we conclude with open problems.

\section{Preliminaries}\label{preliminaries}

Let $\I$ be the set of nonempty, nondense, open connected intervals of the unit circle $S^1= \{z \in \CC : |z|=1\}$.
A \textbf{M\"obius covariant net} is a map $\a$ which assigns to every interval of the circle $I\in\I$ a von Neumann algebra $\a(I)$ acting on a fixed Hilbert space $\h$  satisfying the following properties:
\begin{enumerate}
\item {\sc Isotony:} if $I_1,I_2\in\I$ and $I_1\subset I_2$, then  $\a(I_1)\subset\a (I_2)$;
\item {\sc M\"obius covariance:} there exists a strongly continuous, unitary representation $U$ of
the  M\"obius group $\mob$ ($\simeq \mathrm{PSL(2,\RR)}$) on $\h$ such that $$U(g)\a(I)U(g)^*=\a(gI),\qquad I\in\I, \, g\in\mob;$$
\item {\sc Positivity of the energy:}  the \textbf{conformal Hamiltonian} $L_0$, i.e.\! the generator of the rotation one-parameter subgroup 
has a non negative spectrum. 
\item {\sc Existence and uniqueness of the vacuum:} 
there exists a unique (up to a phase) unit $U$-invariant vector $\Omega\in\h$, i.e.\! $U(g)\Omega=\Omega$ for $g \in \mob$;
\item {\sc Cyclicity:} $\Omega$ is cyclic for the von Neumann algebra $\bigvee_{I\in\I}A(I)$.
\item {\sc Locality:}  if $I_1,I_2\in\I$ and $I_1\cap I_2=\emptyset$, then $\a(I_1)\subset\a(I_2)'$.
\end{enumerate}

We will denote a M\"obius covariant net with the triple $(\a,U,\Omega)$. 
Some consequences of the axioms are (see e.g. \cite{FrG,FJ,GL}): 
\begin{itemize}
\item[7.] {\sc Reeh-Schlieder property}: $\Omega$ is a cyclic and separating vector for each $\a(I)$, $I\in\I$; 
\item[8.]{\sc Haag duality}: $\a(I')'=\a(I)$, where  $I\in\I$ and $I'$ is the interior of $S^1\backslash I$; 
\item[9.]{\sc Bisognano-Wichmann property}: $U(\delta_I(-2\pi t))=\Delta^{it}_{\a(I),\Omega}$ where $\delta_I$ is the dilation subgroup associated to the interval $I$ and $\Delta_{\a(I),\Omega}^{it}$ is the modular group of $\a(I)$ with respect to $\Omega$; 
\item[10.]{\sc  Irreducibility}: $\bigvee_{i\in\I}\a(I)=\b(\h)$;  
\item[11.]{\sc Factoriality}: algebras $\a(I)$ are type $\mathrm{III}_1$ factors;
\item[12.] {\sc Additivity}: let $\{I_\kappa\}\subset\I$  be a covering of $I$, namely $I \subset \bigcup_\kappa I_\kappa$,
then $\a(I)\subset\bigvee_\kappa\a(I_\kappa)$.
\end{itemize}

The following seems relatively less known, yet it follows from M\"obius covariance and
has an important implication \cite[Theorem 1.6]{GLW}.
\begin{itemize}
\item[13.]{\sc Normality and conormality}: for any inclusion $I_1 \subset I_2$,
it holds that $\a(I_1) = \a(I_2)\cap \left(\a(I_1)'\cap \a(I_2)\right)'$ and
$\a(I_2) = \a(I_1)\vee \left(\a(I_1)'\cap\a(I_2)\right)$ .
\end{itemize}
From conormality, it follows that two-interval algebras are factors.
Indeed, take $I_1 \subset I_2$ such that they have no common end points.
Then $I_1$ and $I_2'$ are disjoint intervals with a finite distance.
By Haag duality it follows that $(\A(I_1) \vee \A(I_2'))' = A(I_1)' \cap \A(I_2)$,
and by conormality we have
\begin{align*}
 \left(\A(I_1) \vee \A(I_2')\right) \bigvee \left(\A(I_1) \vee \A(I_2')\right)'
 &= \A(I_1) \vee \A(I_2') \vee \left(\A(I_1)' \cap \A(I_2)\right) \\
 &= \A(I_2) \vee \A(I_2') = \b(\h),
\end{align*}
where the last equality is a consequence of Haag duality and factoriality.
Let us add this to the list of consequences.
\begin{itemize}
\item[14.]{\sc Factoriality of two-interval algebras}: for disjoint intervals $I_1$ and $I_2$ with a finite distance,
$\a(I_1)\vee\a(I_2)$ is a factor.
\end{itemize}

Now, we briefly discuss diffeomorphism covariance. Let $\diff$ be the group of orientation preserving diffeomorphisms of the circle. It is an infinite dimensional Lie group modelled on the real topological vector space $\Vect(S^1)$ of smooth real vector fields on $S^1$ with the $C^\infty$-topology \cite{Milnor}. 
Its Lie algebra has to be considered with the negative of the usual bracket on vector fields, in order to have the proper exponentiation of vector fields. We shall identify the vector field $f(e^{i\theta})\frac d{d\theta}\in \Vect(S^1)$ with the corresponding real function $f\in C^\infty(S^1,\RR)$. 
We  denote with $\diffI$ the subgroup of $\diff$ acting identically on $I'$,
namely the diffeomorphisms of $S^1$ with support included in $\overline I$.

A strongly continuous, \textbf{projective unitary representation} $U$ of $\diff$ on a Hilbert space $\h$ is a strongly continuous homomorphism of $\diff$ into $\U(\h)\slash\T$, the quotient of the group of unitaries in $\b(\h)$ by $\T$. %
The restriction of $U$ to $\mob\subset\diff$ always lifts to a unique strongly continuous unitary representation
of the universal covering group $\widetilde{\mob}$ of $\mob$. $U$ is said to have \textbf{positive energy},
if the generator $L_0$ of rotations, the conformal Hamiltonian, has a nonnegative spectrum in this lift. 
Let $\gamma\in\diff$. 
Note that expressions $\Ad U(\gamma)$ makes sense as an action on $\b(\h)$.
We also write $U(\gamma) \in \M$ although $U(\gamma)$ is defined only up to a scalar.

When one has a strongly continuous projective unitary representation $U$ of $\diff$ with positive energy,
one can differentiate it to obtain the Lie algebra \cite[Appendix A]{Carpi2} (see also \cite{Loke}).
Any smooth function $f\in C^\infty(S^1,\RR)$, as vector fields on $S^1$,
defines the one parameter group of diffeomorphism $\RR\ni t\mapsto \gamma_t \dot= \exp (tf)\in\diff$,
hence, up to an additive constant, defines the self-adjoint generator $T(f)$ of the unitary group
$t\mapsto U(\gamma_t) $. For any real smooth function $f$ as above, $T(f)$ is essentially self-adjoint
on the set $C^\infty(L_0) := \bigcap_{n\in\NN_0} \dom(L_0^n)$.
$T$ shall be called the \textbf{stress energy tensor}.

Irreducible, projective, unitary positive energy representation of $\diff$ are labelled by certain values of the \textbf{central charge} $c>0$ and the \textbf{lowest weight} $h\geq0$. $h$ is the lowest point in the discrete spectrum of the conformal Hamiltonian $L_0$. 
There is a unique (up to a phase) vector $\Phi\in\h$ corresponding to the lowest eigenvalue. See \cite{GoWa, kac} for a detailed description of such representations.

One considers particular elements $\{L_n:n \in\mathbb{Z}\}$,
$L_n = iT(y_n) - T(x_n), L_{-n} = iT(y_n) + T(x_n)$ for $n \in \mathbb{N}$, where $x_n(\theta) := -\sin n\theta$ and $y_n(\theta) := - \cos n\theta$
(there is a canonical way to fix the scalar part of $T(x_n), T(y_n)$, as $L_n, L_{-n}$ and $L_0$
generate a (projective) representation of $\widetilde{\mob}$).
These operators satisfy the so-called Virasoro algebra on the linear span $\d_{\mathrm{fin}}$ of the eigenspaces of $L_0$.
In particular for all $n,m\in\ZZ$: $\d_{\mathrm{fin}}$ is an invariant common core for any closed operator $L_n$; if $n>0$ then $L_n\Phi=0$; $L_{-n}\subset L_n^*$; the family $\{L_n\}_{n\in\ZZ}$ satisfies the Virasoro algebra relations on $\d_{\mathrm{fin}}$:
\[
[L_n,L_m]=(n-m)L_{n+m}+\frac{c}{12}(n^3-n)\delta_{-m,n}{\1}. 
\]

Let $f\in C^{\infty}(S^1,\RR)$ be a  vector field on $S^1$, with Fourier coefficients
$$\hat f_n=\frac1{2\pi}\int_{0}^{2\pi}f(\theta)e^{-in\theta}d\theta,\qquad n\in\ZZ,$$
then, one can recover the stress-energy tensor by
\begin{equation}\label{eq:set}
T(f) = \overline{\sum_{n\in \ZZ} \hat{f}_n L_n}
\end{equation}
and
\[
e^{iT(f)}=U(\exp(f)) 
\]
gives the correspondence between the infinitesimal generators and the representation of $\diff$ (up to a scalar).

Throughout the next few sections we shall often consider the net of von Neumann algebras
\begin{equation*}
\a_U(I)=\{e^{iT(f)}|\, f\in C^\infty(S^1,\RR),\, {\rm supp}(f)\subset I\}''
\qquad (I\in \I).
\end{equation*}
Note that when $U$ is a so-called vacuum representation associated to central charge $c$, $\a_U$ is nothing else than the well-known Virasoro net with central charge $c$.
$\a_U(I_1)$ commutes with $\a_U(I_2)$ if $I_1 \cap I_2 = \emptyset$.

The stress energy tensor can be evaluated on a larger set of functions \cite{CW}. For a continuous
function $f:S^1\to \RR$ with Fourier coefficients $\{\hat{f}_n\}_{n\in\ZZ}$ we shall set
\begin{equation*}
\|f\|_{\frac{3}{2}}= \sum_{n\in \ZZ}|\hat{f}_n|\left(1+ |n|^\frac32\right).
\end{equation*}
Then $\|\cdot\|_{\frac{3}{2}}$ is a norm on the space 
$\{f\in C(S^1,\RR)| \, \|f\|_{\frac{3}{2}}<\infty \}$.
By \cite{CW}, if $f\in C(S^1,\RR)$ with $\|f\|_{\frac{3}{2}}<\infty$, then $T(f)$,
 defined as in \eqref{eq:set}, is self-adjoint and moreover if $f_k\to f$ in the norm $\|\cdot\|_{\frac{3}{2}}$,
then $T(f_k)\to T(f)$ in the strong resolvent sense. In particular,
even for a non necessarily smooth function $f$ with $\|f\|_{\frac{3}{2}} < \infty, \supp f  \subset I$,
the self-adjoint $T(f)$ is still affiliated to $\a_U(I)$.

 
We shall say that  a M\"obius covariant net  $(\a,U,\Omega)$ is \textbf{conformal} (or \textbf{diffeomorphism covariant}) if the $\mob$ representation $U$ extends to a projective unitary representation $\diff\to \U(\h)\slash\T$ of $\diff$ (that with a little abuse of notation we continue to indicate the extension with $U$) and satisfying 

\begin{itemize}
\item $ \Ad U(\gamma)(\a(I)) = \a(\gamma I),\; \text{ for } \gamma\in\diff$
\item $ \Ad U(\gamma)(x) = x, \; \text{ for } \gamma\in\diffI, x\in\a(I')$
\end{itemize}

Even for a non necessarily smooth function $f$ with $\|f\|_{\frac{3}{2}} < \infty, \supp f  \subset I$,
the self-adjoint $T(f)$ is still affiliated to $\a_U(I)$.
 
Now we recall the definition of the split property for von Neumann algebra inclusions and conformal nets.
\begin{definition}\label{def:split}
Let $(\N\subset \M,\Omega)$ be an \textbf{standard inclusion} of von Neumann algebras, i.e.\! $\Omega$ is a cyclic and separating vector for $N,\, M$ and $N'\cap M$.

A standard inclusion $(\N\subset \M,\Omega)$ is \textbf{split} if there exists a type I factor $\R$ such that $\N\subset \R \subset \M$.

A M\"obius covariant net $(\a,U,\Omega)$ satisfies the \textbf{split property} if  the von Neumann algebra inclusion $\a(I_1)\subset\a(I_2)$ is split, for any inclusion of intervals $I_1\Subset I_2$,
namely when $I_1$ and $I_2$ have no common end points.
\end{definition}
The following proposition provides an equivalent condition to the split property.
 Although similar statements are quite well-known to experts (see \cite{DL83} and \cite[Below Definition 1.4]{DL}),
 the precise assumptions we need are difficult to find in the literature
 (note, for example, that we do not assume neither the separability of the underlying Hilbert space
 \footnote{If the Hilbert spaces are not separable, several well-known statements no longer hold.
 For example, an isomorphism between type III algebras might be not a unitary equivalence.}
 nor the faithfulness of the split state in the implication $2 \Rightarrow 1$ below).
\begin{proposition}\label{prop:split}
Let $(\N\subset \M,\Omega)$ be a standard inclusion of von Neumann algebras.
We further assume that that $\N\vee \M'$ is a factor.
Then the following are equivalent.
\begin{enumerate}
\item $\N\subset \M$ is split; 
\item there exists a normal state $\phi$ on $\N\vee \M'$ such that
the restrictions $\phi_\N$ and $\phi_{\M'}$ are faithful and
$\phi$ is split, namely,
\[
 \phi(xy)=\phi(x)\phi(y),\qquad x\in \N,y\in \M'.
\]
\end{enumerate}
\end{proposition}
\begin{proof}
 If $\N \subset \M$ is split, namely if there is an intermediate type I factor $\R \simeq \b(\mathcal{K})$,
 then $\N \vee \M'$ is isomorphic to $\N\otimes \M'$, from which the implication $1 \Rightarrow 2$ follows.

Conversely, let there be a split state as in $2$.
 First of all, as $\phi$ are faithful on $\N$ and $\M'$, their GNS representations $\pi_{\N}, \pi_{\M'}$ are faithful and have
 a cyclic and separating vector.
 Next, as $\phi$ is normal on $\N\vee \M'$, its GNS representation $\pi_{\N\vee \M'}$ is also normal. The Hilbert space supporting $\pi_{\N\vee \M'}$ is isomorphic to the closure of $\N\vee_{\mathrm{alg}}\M'$ w.r.t.\! the scalar product inherited by the normal state $\phi$ as $\langle x,y\rangle_\phi=\phi(x^*y)$. By the factorization assumption on $\phi$, the Hilbert space is the tensor product  $L^2(\N,\langle\cdot,\cdot\rangle_\phi)\otimes L^2(\M',\langle\cdot,\cdot\rangle_\phi)$ and the GNS representation $\pi_{\N\vee\M'}$ restricted to $\N$ and $\M'$ are of the form $\pi_\N\otimes \1$ and $\1\otimes \pi_{\M'}$, respectively. Furthermore, as both $\N$ and $\M'$ have a cyclic and separating vector $\Omega$, their GNS representations $\pi_{\N}, \pi_{\M'}$ are actually unitary equivalences \cite[Corollary 10.15]{alapkonyv}. As a consequence, by normality, we can assume that $\pi_{\N\vee \M'}(\N\vee \M') = \N\otimes \M'$.
%
 Furthermore, by assumption $\N \vee \M'$ is a factor, hence the GNS representation is an isomorphism.
 Now, both $\N\vee \M'$ and $\N\otimes \M'$ have a cyclic and separating vector
 ($\Omega$ and $\Omega\otimes \Omega$ respectively), therefore, the GNS representation is actually a unitary equivalence.
 Then the preimage $\R = \pi_{\N\vee \M'}^{-1}\left(\b(\h)\otimes \CC\1\right)$ gives the intermediate subfactor $\N \subset \R \subset \M$.
%
\end{proof}

\begin{remark}
The split property implies separability of the Hilbert space.
Indeed, if we have a standard split inclusion of von Neumann algebra on an Hilbert space $\h$, then $\h$ has to be separable:
$\Omega$ is a cyclic and separating vector for the intermediate type I factor $\R$.
By considering the cardinality of the basis, either $\R$ or $\R'$ must be isomorphic to $\b(\h)$ and
$\Omega$ defines a faithful vector state on it, hence $\b(\h)$ is $\s$-finite, which is only possible if $\h$ is separable. 
\end{remark}

\section{Local decompositions of \texorpdfstring{$e^{-\beta L_0}$}{e}}\label{localdecomp}

Let $U$ be a strongly continuous projective unitary representation $U$ of $\diff$ with positive energy
which extends a proper representation of $\mob$ \footnote{We shall mean by ``proper'' that
an object or a relation is defined including the phase, hence not only being projective.}.
In what follows, for a $\beta>0$, $r=e^{-\beta}$ and two open proper arcs (intervals) $I_c,I_d\in \I$ that cover the circle: $I_c\cup I_d= S^1$,
we shall find a decomposition
$e^{-\beta L_0} = r^{L_0} = C_r D_r$ with the bounded operators $C_r\in \a_U(I_c)$ and $D_r\in \a_U(I_d)$.
The main idea for producing such a decomposition was already
presented and exploited in \cite{weiner} and in the proof of \cite[Theorem 2.16]{CCHW}. Here we shall recall the essential points of the argument presented there and then adjust and refine it to our purposes. 
\begin{proposition}
\label{prop:dec1}
Let $I_c,I_d\in \I$ be two open proper arcs covering the circle: $I_c\cup I_d = S^1$.
Then there exist two norm-continuous families of operators 
$(0,1)\ni r\mapsto C_r \in\a_U(I_c)$ and
$(0,1)\ni r\mapsto D_r \in\a_U(I_d)$ such that
$$
r^{L_0}=C_r D_r \;\;\;\;\textrm{and}\;\;\;\; \|C_r\|,\|D_r\|\leq \frac{1}{r^q}
$$
where the exponent $q = \frac{c}{48}(N^2-1)$ with $N$ being a positive integer such that
$6\pi/N$ is smaller than the lengths of both arcs that are obtained by taking the intersection $I_c\cap I_d$
(note that $N$ must be at least $4$).  
\end{proposition}
\begin{proof}
\begin{figure}[ht]
\centering
\begin{tikzpicture}[line cap=round,line join=round,>=triangle 45,x=1.0cm,y=1.0cm, scale = 0.8]
\clip(-5,-5) rectangle (5,5);
\draw [line width=3pt] (0,0) circle (2.5cm);
\draw  (0.694592711,3.939231012) arc (80:300:4) ;
\draw  (-0.392200842,-4.482876141) arc (265:480:4.5) ;
\draw  (0.278898377,3.187823034) arc (85:95:3.2) ;
\draw  (0,3.6) arc (90:100:3.6) ;

\draw (-3,3) node[anchor=north west] {$I_c$};
\draw (2.9,2.9) node[anchor=north west] {$I_d$};
\draw (0.33,3.35) node[anchor=north west] {$ I_k $};
\draw (0,3.94) node[anchor=north west] {$\tilde I_k$};
\end{tikzpicture}\caption{Intervals $I_c, I_d$ covering $S^1$ and $I_k, \tilde I_k$ with $N = 36$.}
\label{fig:decomp1}
\end{figure}
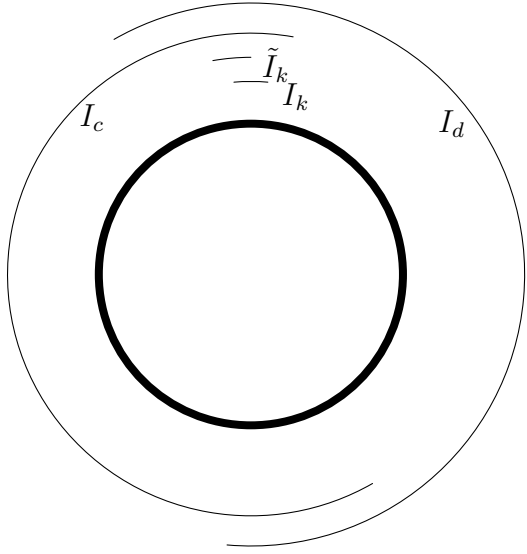

Let us fix a positive integer $N$ satisfying the condition of the proposition (see Figure \ref{fig:decomp1}).
The operators $H := \frac{1}{N}L_0+\frac{c}{24}(N-\frac{1}{N}) \1$, $L_+ := \frac{1}{N}L_{-N}$ and $L_- := \frac{1}{N}L_N$ satisfy the relations
\begin{equation*}
[H,L_\pm ] = \mp L_\pm, \;\;\;\; [L_-,L_+] = 2H,\;\;\;\; L_\pm\subset L_\mp^*.
\end{equation*}
Moreover, $H$ is diagonalizable with non-negative eigenvalues only, $L_\pm$ is defined on the span of the eigenvectors of $H$
which is also an invariant subspace for these operators. It then follows that these operators generate a strongly continuous, positive energy
unitary representation of the universal cover $\widetilde{\mob}$ of the M\"obius group. 
This construction --- both at the Lie algebra as well as the Lie group level --- was already considered and used by various authors;
see e.g.\! the work \cite{LX}.
In particular,
\begin{equation*}
P=\frac{1}{4}\overline{(2H-L_+ - L_-)}\;\;\; \textrm{and} \;\;\; \tilde{P}=\frac{1}{4}\overline{(2H+L_+ + L_-)}
\end{equation*} 
are conjugate to each other by the unitary operator
$e^{i\pi H}$,
with $P$ being the
self-adjoint generator of ``translations'' with spectrum ${\rm Sp}(P)= {\rm Sp}(\tilde{P})=\RR_+ \cup \{0\}$.
Moreover, by \cite[Theorem 3.3]{BDL} we have the relation
\begin{equation}\label{BDL_formula}
e^{-2sH} = e^{-{\rm tanh}(\frac{s}{2})P}e^{-{\rm sinh}(s)\tilde{P}}e^{-{\rm tanh}(\frac{s}{2})P}
\end{equation}
for all $s>0$. Let us now consider how $P$ and $\tilde P$ can be written in terms of the stress-energy $T$. We have
\begin{equation*}
P = \frac{1}{4N}\overline{(2L_0 - L_{-N}-L_N)} + \frac{c}{48}\left(N-\frac{1}{N}\right)\1 = T(p)+ b\1
\end{equation*}
and likewise
$\tilde{P} = T(\tilde{p})+ b\1$,
where 
\begin{equation*}
b = \frac{c}{48}\left(N-\frac{1}{N}\right)
\end{equation*}
and 
$p$ and $\tilde{p}$ are the functions defined by the formulas
$p(z)=\frac{1}{4N}(2-z^N-z^{-N})$ and
 $\tilde{p}(z)=\frac{1}{4N}(2+z^N+z^{-N})$. 
 
The function $p$ is nonnegative on $S^1$ and it has exactly $N$ points where its value is zero:
\begin{equation*}
p(z) = 0 \;\; \Longleftrightarrow  \;\; z = e^{i\frac{2\pi}{N} k}\;\; {\rm for}\;k=1,\ldots N. 
\end{equation*}
All these null-points are of course local (and also global) minima, where the derivative is zero.
We can thus ``cut'' $p$ into $N$ ``nice'' pieces: $p = p_1 + \ldots + p_N$ where the
support of the nonnegative function $p_k$ is the closure of the arc 
\begin{equation*}
I_k = \left\{e^{i\theta}: \frac{k-1}{N}<\frac{\theta}{2\pi}<\frac{k}{N} \right\},
\end{equation*}
and $\|p_k\|_{\frac{3}{2}}<\infty$. This latter follows from the fact that $p_k$ is once differentiable
and its derivative is of bounded variations; see the similar considerations at \cite[Lemma 5.3]{CW}. Thus for every $k=1,\ldots N$, 
\begin{equation*}
P_k = T(p_k) + \frac{b}{N}\1
\end{equation*}
is a well-defined self-adjoint operator affiliated to $\a_U(I_k)$
and we have 
$P=\overline{P_1 + \ldots +P_N}$. Since the terms in this 
decomposition are affiliated to commuting factors, just as in the proof \cite[Proposition 3.2]{weiner}, we have that 
\begin{equation*}
{\rm Sp}(P_1) + \ldots + {\rm Sp}(P_N)= {\rm Sp}(P) = \RR_+\cup\{0\}.
\end{equation*}
On the other hand, the spectrum of the operators $P_k$ $(k=1,\ldots N)$ must all
coincide, since using rotations one can easily show that they are all unitary conjugate to each other.
It then follows that each of them must be a positive operator. Thus for the bounded operator $e^{-{\rm tanh}(\frac{s}{2})P}$ appearing in formula \eqref{BDL_formula}, we have the {\it local} decomposition into a product of commuting bounded operators
\begin{equation*}
e^{-{\rm tanh}(\frac{s}{2})P} = \prod_{k=1}^N e^{-{\rm tanh}(\frac{s}{2})P_k}
\end{equation*}
where the norm of each term is smaller or equal than $1$.

Let us turn to $\tilde{P}$. As we have $\tilde{P} = \Ad e^{i\frac \pi N L_0}(P)$,
the {\it localization} of $\tilde{P}_k = \Ad e^{i\frac \pi N L_0}(P_k)$ are different from that of $P$: $\tilde P_k$ is affiliated to 
$\a_U(\tilde{I}_k)$ where $\tilde{I}_k = e^{i\frac \pi N} I_k$
and $\a_U$ is defined in Section \ref{preliminaries}
(we are considering the intervals as subsets in $\CC$).
With this localization, we can still assure the strong commutation between $P_k$ and $\tilde{P}_j$ whenever $k\neq j,j+1$ (mod $N$). 
So in the decomposition
\begin{eqnarray*}
\nonumber
e^{-2sH} &=& e^{-{\rm tanh}(\frac{s}{2})P}\;\; e^{-{\rm sinh}(s)\tilde{P}}\;\; e^{-{\rm tanh}(\frac{s}{2})P} \\
&=& \left(\prod_{k=1}^N e^{-{\rm tanh}(\frac{s}{2})P_k}\right) \left(\prod_{k=1}^N e^{-{\rm sinh}(s)\tilde{P}_k}\right)
\left(\prod_{k=1}^N e^{-{\rm tanh}(\frac{s}{2})\tilde{P}_k}\right)
\end{eqnarray*}
we can make some rearrangements. Note that $e^{-2sH} = r^{-L_0}r^{2q}$, where $q = \frac c{48}(N^2-1)$ if we set
$r=e^{-2s/N}$. To shorten notations, let us introduce the self-adjoint contractions
$X_k=e^{-{\rm tanh}(\frac{s}{2})P_k}$ and $Y_k=e^{-{\rm sinh}(s)\tilde{P}_k}$. For simplicity, we did not indicate their dependence on $r$, but note that in the range $0<r<1$ they depend norm-continuously on $r$
(for $t > 0, x \ge 0$, the function $e^{-tx}$ is uniformly continuous in $t$).

All $X$-operators and separately, all $Y$-operators commute between themselves, and moreover
$[X_l,Y_m]=0$ whenever $l\neq m, m+1$ (mod $N$).
Recall that $\frac {6\pi} N$ is smaller than the length of each of the intervals of $I_c\cap I_d$.
By cyclically renaming the intervals (but keeping the relation between $I_k$ and $\tilde I_k$ and
the corresponding localization of the operators), we may assume that 
there are $1 \le k < j \le N$ such that $I_k\cup I_{k+1}$ and $I_j\cup I_{j+1}$
are included in the different connected components of $I_c\cap I_d$.
Furthermore, to fix the notation, we may assume that $I_k\cup\cdots \cup I_{j+1} \subset I_c$, while
$I_j\cup\cdots \cup I_N\cup I_1\cdots \cup I_k \subset I_d$.
Note that $\tilde I_k\cup\cdots \cup \tilde I_{j} \subset I_c$ and
$\tilde I_j\cup\cdots \cup I_N\cup I_1\cdots \cup \tilde I_{k-1} \subset I_d$
(see Figure \ref{fig:decomp1bis}).

\begin{figure}[ht]
\centering
\begin{tikzpicture}[line cap=round,line join=round,>=triangle 45,x=1.0cm,y=1.0cm, scale = 0.6]
\clip(-7,-7) rectangle (7,7);
\draw [line width=4pt] (0,0) circle (2.5cm);
\draw  (1.128713155,6.401250395) arc (80:300:6.5) ;
\draw  (-1.041889066,-5.908846518) arc (260:480:6) ;
\draw  [line width=3pt] (0.278898377,3.187823034) arc (85:95:3.2) ;
\draw  [line width=1pt] (0.278898377,3.187823034) arc (85:275:3.2) ;
\draw  [line width=3pt] (0.278898377,-3.187823034) arc (275:285:3.2) ;

\draw  [line width=3pt] (0,4) arc (90:100:4) ;
\draw  [line width=1pt] (0,4) arc (90:270:4) ;
\draw  [line width=3pt] (0,-4) arc (270:280:4) ;

\draw  [line width=3pt] (-0.435778714,4.98097349) arc (95:105:5) ;
\draw  [line width=1pt] (-0.435778714,4.98097349) arc (95:265:5) ;
\draw  [line width=3pt] (-0.435778714,-4.98097349) arc (265:275:5) ;

\draw (-4.8,4.8) node[anchor=north west] {$I_c$};
\draw (3.2,4.8) node[anchor=north west] {$I_d$};
\draw (0.33,3.35) node[anchor=north west] {$ I_k $};
\draw (0.18,4.5) node[anchor=north west] {$\tilde I_k$};
\draw (-0.2,5.4) node[anchor=north west] {$ I_{k+1} $};
\draw (1,-2.65) node[anchor=north west] {$ I_{j+1} $};
\draw (0.88,-3.5) node[anchor=north west] {$\tilde I_j$};
\draw (0.7,-4.6) node[anchor=north west] {$ I_j $};
\end{tikzpicture}\caption{Localization of the factors of $C_r$.
The indicated intervals $I_\bullet, \tilde I_\bullet$ correspond to thick segments.
The operators $\prod_{l=k}^{j+1} X_l,\, \prod_{l=k}^{j} Y_l,\,\prod_{l=k+1}^{j} X_l$
are localized in the arcs, from the inside, respectively.
The corresponding factors in $D_r$ are localized in the complements of these arcs, respectively.
}
\label{fig:decomp1bis}
\end{figure}
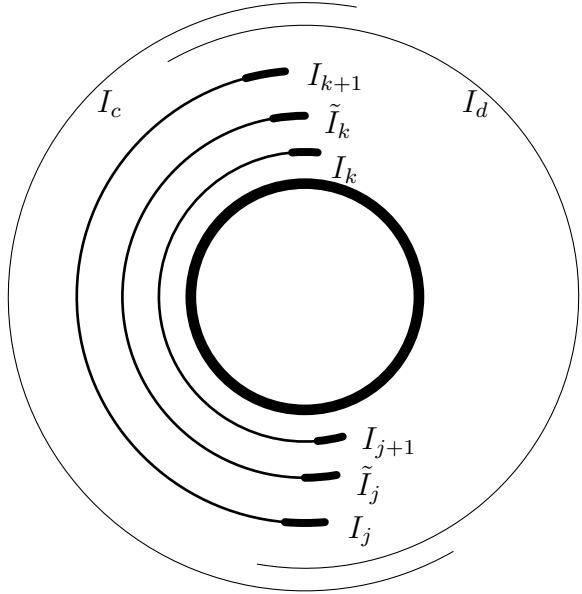

By the localization explained above, we obtain
\begin{eqnarray*}
\nonumber
r^{-L_0}r^{2q}
&=& \left(\prod_{l=1}^N X_l\right) \;\;\; \left(\prod_{l=1}^N Y_l\right)\;\;\;
\left(\prod_{l=1}^N X_l\right) \\
\nonumber
&=&
\left(\prod_{l=1}^{k-1} X_l\prod_{l=k}^{j+1} X_l\prod_{l=j+2}^{N} X_l\right) 
 \, \left(\prod_{l=1}^{k-1} Y_l\prod_{l=k}^{j} Y_l\prod_{l=j+1}^{N} Y_l\right) \,
\left(\prod_{l=1}^{k} X_l\prod_{l=k+1}^{j} X_l\prod_{l=j+1}^{N} X_l\right)  \\
&=&
\left(\prod_{l=k}^{j+1} X_l\prod_{l=k}^{j} Y_l\prod_{l=k+1}^{j} X_l\right)\;
\left(\prod_{l=1}^{k-1} X_l\prod_{l=j+2}^{N} X_l \prod_{l=1}^{k-1} Y_l \prod_{l=j+1}^{N} Y_l
\prod_{l=1}^{k} X_l \prod_{l=j+1}^{N} X_l\right).
\end{eqnarray*}
Here the first part $C_r=\left(\prod_{l=k}^{j+1} X_l\prod_{l=k}^{j} Y_l\prod_{l=k+1}^{j} X_l\right)$ is an element of $\a_U(I_c)$, where 
whereas the second part $D_r=\left(\prod_{l=1}^{k-1} X_l\prod_{l=j+2}^{N} X_l \prod_{l=1}^{k-1} Y_l \prod_{l=j+1}^{N} Y_l
\prod_{l=1}^{k} X_l \prod_{l=j+1}^{N} X_l\right)$ 
is an element of $\a_U(I_d)$.

By construction, $\|C_r\|,\|D_r\|\leq 1$. Thus, we have obtained the desired decomposition 
$r^{-L_0} =  (\frac{1}{r^q} C_r) (\frac{1}{r^q} D_r)$.

\end{proof}

In the above proposition we specifically worked with $L_0$. However, by considering the adjoint actions of $U(\gamma)$ for all diffeomorphisms $\gamma\in \diff$ on the 
decompositions found above, it is now easy to draw the following conclusion.
\begin{corollary}
\label{rt_decomp}
Let $I_c,I_d\in \I$ be two open proper arcs such that $I_c\cup I_d = S^1$, and $f$ a strictly positive smooth function on $S^1$. Then there exist two norm-continuous 
families of operators 
$(0,1)\ni r\mapsto C_r \in\a_U(I_c)$ and
$(0,1)\ni r\mapsto D_r \in\a_U(I_d)$ such that
$r^{T(f)}=C_r D_r.$
\end{corollary}

\section{Further decomposition}\label{furtherdecomp}

In this section, we shall consider further decompositions of $r^{L_0}$.
For this purpose, it will be important that some representations of a real Lie group
contained in a complex Lie group can be continued holomorphically to
representations of a certain complex Lie semigroup \cite{Ols}.
This applies to positive energy representations of the M\"obius group and the semigroup can be explicitly constructed.

We denote the (real) Lie algebra of $\mob$ by $\mathfrak{M\ddot ob}$
and its generators by $il_0, i(l_1 + l_{-1}), l_1-l_{-1}$ (while $L_0, L_1+L_{-1},L_1- L_{-1}$ are reserved for representations).
Let us introduce an invariant cone
\[
 \C := \mathrm{Conv}\{ \Ad_g(r\cdot il_0): g \in \mob, r > 0\},
\]
where $\mathrm{Conv}$ stands for the convex hull. This is invariant under the adjoint action
by $\mob$ by definition, and is nontrivial (namely, not the whole $\mathfrak{M\ddot ob}$,
because all these elements are represented by a positive
operator in a positive energy representation of $\mob$, and we know that such nontrivial positive energy representations
exist). 

Let $\mob_\CC$ denote the group of the complex fractional linear transformations $\mob_\CC$,
which is a complex Lie group which includes $\mob$ as a real Lie subgroup.
They act naturally on the Riemann surface $P(\CC)$ by
$z \mapsto \frac{az + b}{cz + d}$, and $\mob_\CC$ is identified with $\mathrm{PGL}(2,\CC)$, while
$\mob$ with the subgroup of the form $\left(\begin{array}{cc}a & b\\ \bar b & \bar a \end{array}\right)$
with $|a|^2-|b|^2 = 1$. Recall that $\mob$ preserves the open unit disk $D_1$,
and moreover maps the unit circle $S^1$ onto itself.

We consider the semigroup $\Gamma(\C)$ of $\mob_\CC$ of the elements which map $\overline{D_1}$ into $D_1$.
It contains the contractions $\{\kappa_t:z\mapsto e^{-t}z\,|\,t > 0\}$.
We have a convenient representation of elements in $\Gamma(\C)$.
\begin{lemma}\label{lem:decsemi}
Every $\gamma \in \Gamma(\C)$ has the form
\begin{equation}\label{eq:decsemi}
\gamma=g_1\cdot\kappa_t\cdot g_2
\end{equation}
where $g_1,g_2\in\mob$ and $t > 0$.
For a given $\gamma$, $g_1$ and $g_2$ are uniquely determined up to a rotation.
\end{lemma}
\begin{proof}
Let $\gamma\in\Gamma(\C)$. By definition, $\gamma(\overline{D_1})\subset D_1$ and $\gamma(D_1)$ is a disk
because $\gamma$ is a linear fractional transformation.
There exists an element in $\mob$ which maps $\gamma (D_1)$ to a disk concentric to $D_1$.
Indeed, up to a rotation, we can assume that the diameter of $\gamma (D_1)$ is included in $(-1,1)$.
Then there is a (unique) dilation $g_1 \in \mob$ corresponding to the upper half-circle
(see \cite[Appendix A]{weinerthesis}) such that $g_1^{-1}$ maps the diameter of $\gamma(D_1)$ onto
a symmetric interval $(-s,s)$, $0 < s < 1$, thus $(g_1^{-1} \cdot \gamma) (D_1)$ is concentric to $D_1$.
Now, there exists a $t > 0$ such that $\kappa_t(D_1)=(g_1^{-1} \cdot \gamma) (D_1)$.
As $g_2 := \kappa_{-t}\cdot g_1^{-1}\cdot \gamma$
is a linear fractional transformation which preserves $S^1$, it must be in $\mob$.

As for uniqueness, one only has to note that rotations are the only elements in $\mob$
which preserve a contracted circle $rS^1, 0 < r < 1$.
\end{proof}

By equation \eqref{eq:decsemi}, an element of $\gamma\in\Gamma(\C)$  can be uniquely decomposed as
$\gamma=g\cdot \tilde\kappa_t$ where  $\tilde\kappa_t=\tilde g\cdot\kappa_t\cdot\tilde g^{-1}$ for some $g,\tilde g\in\mob$
($\tilde g$ is unique up to a rotation).
In addition, $t\mapsto\tilde \kappa_t = \Exp (t\, y)$ is the one-parameter semigroup corresponding to an element $y \in \C$.

The main steps of the following result are due to Olshanskii \cite[Theorem 4.5]{Ols},
but the original proof is written
for the universal covering semigroup,
and in order to obtain the result we need, one would have to digest some issues on the topology and notations.
For better readability, we present the proof in our case.

\begin{theorem}\label{thm:ols}
Any unitary, strongly continuous, positive energy representation of the group $\mob$ admits
a unique continuous extension to the closure $\overline{\Gamma(\C)}$ which is analytic on the interiour $\Gamma(\C)$.
\end{theorem}
\begin{proof}
Let $\gamma\in\Gamma(\C)$ and $V$ be a positive energy representation of $\mob$.
Then there is the corresponding representation of the Lie algebra $\mathfrak{M\ddot ob}$
(on the G\aa rding domain).
Let $L_0$ be the conformal Hamiltonian in the representation $V$ of $\mob$, corresponding to the element $l_0$
($t\mapsto e^{itL_0}$ is the unitary representation of rotations, as usual).
With the decomposition $\g = g \cdot \Exp (t \,y)$ of Lemma \ref{lem:decsemi}
and the remark thereafter,
we set $\tilde V(\gamma)=V(g)e^{-tY}$, where $Y$ is the representation of $y$ in $ V$.
It is well-defined as the difference in rotation does not matter,
and $e^{-t Y}$ is a positive contraction.


By Nelson's theorem \cite[Corollary 3.2]{Nel} and the commutation relations in $\mathfrak{M\ddot ob}$,
there is a dense set of analytic vectors $\xi$ of the representation $V$
(actually one can simply consider linear combinations of eigenvectors of $L_0$)
such that the map $\mob \ni g\mapsto V(g)\xi$ continues analytically to a neighborhood of the identity $W$ in $\mob_\CC$,
where $g$ is considered as an element in $\mob_\CC$.

We claim that, for an element $\gamma \in W \cap \Gamma(\C)$, the map $V(\gamma):\h \ni \xi \mapsto V(\gamma)\xi \in \h$ is bounded.
Every $g \in \mob$ can be uniquely written as $\underline g \rho_\theta$,
where $\underline g$ belongs to the translation-dilation subgroup and $\rho_\theta$ is a rotation
(by the Iwasawa decomposition).
Now, on one hand, if one continues analytically
the rotations $\left(\begin{array}{cc} e^{\frac{i\theta} 2} & 0\\ 0 & e^{-\frac{i\theta} 2} \end{array}\right)$
to $\zeta = \theta + i\lambda, \lambda > 0$ in $\mob_\CC$, one obtains $\rho_\theta\kappa_{\lambda} \in \Gamma(\C)$,
where $\kappa_\lambda$ is a contraction.
On the other hand, we have $V(\underline g \rho_\theta) = V(\underline g)e^{i\theta L_0}$.
As we continue $V(\underline g \rho_\theta)$ analytically in $\theta$ to $\theta + i\lambda$,
this is a bounded operator for $\lambda \ge 0$ as $L_0$ has positive spectrum.
For any other element $\gamma = g\tilde \kappa_\lambda$ in $\Gamma(\C)$, there is a subgroup
$\tilde g \kappa_\lambda \tilde g^{-1}$, conjugate to the rotation group,
and we can infer that the continuation $V(\gamma)$ is bounded by the Iwasawa decomposition
with respect to this subgroup.
By the density of analytic vectors, $\Gamma(\C)\cap W \ni \gamma \mapsto V(\gamma) \in \b(\h)$ can be defined
by boundedness, and since the uniform limit of analytic functions is analytic,
it is analytic in $\Gamma(\C)\cap W$ and
continuous on the closure $\overline{\Gamma(\C)}\cap W$ in the strong topology because,
for any $\xi$ in the dense domain of analytic vectors, the map $W \ni \gamma\mapsto V(\gamma)\xi \in \h$
is continuous and the family $V(\gamma)$ is uniformly bounded for $\gamma \in \Gamma(\C)\cap W$.

Furthermore, this analytic continuation coincides with $\tilde V(\gamma)$, when restricted to $W\cap \Gamma(\C)$.
Indeed, again by the Iwasawa decomposition we have $g = \underline g \rho_\theta \in \mob$,
Now, for an element of the form $\gamma = \underline g \rho_\theta \kappa_t$,
we have $\tilde V(\gamma) = V(\underline g)e^{i\theta L_0}e^{-tL_0}
= V(\underline g)e^{i(\theta+it)L_0}$,
and this is indeed an analytic continuation of $V$ in the variable $\theta$,
hence it must coincide with the above continuation.
Similarly, one can prove it for arbitrary element $\g = g \cdot \tilde \kappa_t$
by considering the conjugate Iwasawa decomposition.


The map $\overline{\C}\ni y \mapsto \tilde V(\Exp y)$ is strongly continuous
and real analytic on $\C$.
To see this, note that it is always possible to find an $n\in\NN$ such that $\Exp (-y/n)\in \overline{\Gamma(\C)}\cap W$
and the $n$-th power $\tilde V(\Exp y) = \tilde V(\Exp y/n)^n$ on uniformly bounded sets is strongly continuous
and real analytic.
Now $\tilde V(\g)$ is analytic on $\Gamma(\C)$ and continuous on the closure $\overline{\Gamma(\C)}$.
Indeed, the expression $\tilde V(g\cdot \Exp y) = V(g)e^{-Y}$,
where $Y$ is the representation of $y$ in $V$,
is real analytic in both variables $g$, and $y \in \C$. Furthermore, we have seen that it is complex analytic
in an open set $W\cap \Gamma(\C)$. From this we conclude that it is complex analytic in the whole $\Gamma(\C)$:
the domain of complex analyticity is open by definition, and also closed,
because if there were a boundary point $\gamma$, one could use the real analyticity (convergence of the Taylor expansion)
to continue complex analytically the map to a neighborhood of $\gamma$, but this continuation would have to coincide
with the original map because of the real analyticity.

It remains to prove that $\tilde V$ respects the product relation of the semigroup,
namely that $\tilde V(\gamma_1)\tilde V(\gamma_2)=\tilde V(\gamma_1\gamma_2)$.
This follows by the fact that the maps $\gamma_1\mapsto \tilde V(\gamma_1)\tilde V(\gamma_2)$ ($\gamma_2$ is fixed)
and $\gamma_2\mapsto \tilde V(\gamma_1)\tilde V(\gamma_2)$ ($\gamma_1$ is fixed) are analytic on $\Gamma(\C)$
and coincide on $\mob\times\mob \subset \overline{\Gamma(\C)}\times \overline{\Gamma(\C)}$.
\end{proof}

\begin{corollary}\label{co:conj}
 Let $V$ be a positive energy, strongly continuous, unitary representation of $\mob$ with associated conformal Hamiltonian $L_0$.
Then for every $r\in (0,1)$ there exist $r_1, r_2\in (0,1)$ and $g, g_1, g_2 \in \mob$, $g \neq \mathrm{id}$,
such that
\[
 r^{L_0} = r_1^{H_1}r_2^{H_2}V(g)
\]
in the proper sense, where $H_j = \Ad U(g_k)(L_0)$ ($k = 1,2$).

\end{corollary}
\begin{proof}
 We choose two elements $\tilde g_1, \tilde g_2\in \mob$ such that
 $\tilde H_k = \Ad V(\tilde g_k)(L_0)$, and $\tilde H_1$ and $\tilde H_2$ do not strongly commute:
 such choices are actually abundant, since $L_0$ is maximally abelian in the Lie algebra.

Now, arguing by contradiction, assume that there exist $r_1,r_2\in(0,1)$ such that  $r_1^{\tilde H_1}$ and $r_2^{\tilde H_2}$ commute.
The operator valued functions $z\mapsto e^{z \tilde H_1}r_2^{\tilde H_2}$ and $z\mapsto r_2^{ \tilde H_2}e^{z \tilde H_1}$
are continuous on $\Re z\leq 0$ and analytic in $\Re z< 0$. Then the maps coincide when $z=q \ln r_1$ with $q\in \mathbb R$,
hence by analyticity they must coincide on the full domain.
One can argue analogously with $w\mapsto e^{z\tilde H_1}e^{w \tilde H_2}$ and
$w\mapsto e^{w \tilde H_2}e^{z\tilde H_1}$ to get that $e^{it\tilde H_1}$ and $e^{is\tilde H_2}$ commute for any $s,t\in\RR$ by analytic continuation. This contradicts the fact that  their generators do not strongly commute. 

Let $r_1,r_2\in (0,1)$. By applying Lemma \ref{lem:decsemi} and Theorem \ref{thm:ols} to
\[
 r_1^{\tilde H_1}r_2^{\tilde H_2} = V(\tilde g_1)r_1^{L_0}V(\tilde g_1)^*V(\tilde g_2)r_2^{L_0}V(\tilde g_2)^*, 
\]
we obtain $g_3, \tilde g_3 \in \mob$ and $r \in (0,1)$ such that
 \[
  V(\tilde g_1)r_1^{L_0}V(\tilde g_1)^*V(\tilde g_2)r_2^{L_0}V(\tilde g_2)^* = V(g_3)r^{L_0}V(\tilde g_3)^*,
 \]
 in the proper sense, or equivalently,
 \[
 r^{L_0} = V(g_3^{-1}\tilde g_1)r_1^{L_0}V(g_3^{-1}\tilde g_1)^*V(g_3^{-1}\tilde g_2)r_2^{L_0}V(g_3^{-1}\tilde g_2)^* V(g_3\tilde g_3).
 \]
 By defining $g_k = g_3^{-1}\tilde g_k$, hence accordingly $H_k := \Ad V(g_3^{-1}\tilde g_k)(L_0)$
 and $g := g_3\tilde g_3$, we obtain the desired equality.
 To check that $g \neq \mathrm{id}$, note that by our choice of $\tilde H_k$, $r_1^{H_1}$ and $r_2^{H_2}$ do not commute as well.
 Yet, in the equality
 \[
 r^{L_0} = r_1^{H_1}r_2^{H_2}V(g),
 \]
 the left-hand side is self-adjoint, while if $g = \mathrm{id}$, the right-hand side would not be self-adjoint.
Therefore, $g \neq \mathrm{id}$. Notice as $r_1,r_2$ are chosen arbitrarily in $(0,1)$, then the decomposition holds for any $r\in(0,1)$ just by strong continuity of the (extension of) the representation $V$ to $\Gamma(\C)$.
\end{proof}

\begin{proposition}
\label{prop:dec2}
Let $\A$ be a conformal net, $U$ be the associated projective unitary representation of $\diff$,
and $L_0$ the conformal Hamiltonian.
For every $r\in (0,1)$ there exists a M\"obius transformation $g \neq {\rm id}$, such that whenever $I_c,I_d\in \I$ are
two open proper arcs covering the circle, i.e.\! $I_c\cup I_d = S^1$, we can find two bounded operators $C \in\a_U(I_c)$ and $D \in\a_U(I_d)$ giving the decomposition
\begin{equation*}
r^{L_0} = C D U(g)
\end{equation*}
in the proper sense.
\end{proposition}
\begin{proof}
 We apply Corollary \ref{co:conj} to obtain
 $r, r_1, r_2\in (0,1)$, $g \in \mob, g \neq \mathrm{id}$ and $H_1, H_2$
such that
 \[
  r^{L_0} = r_1^{H_1}r_2^{H_2}U(g).
 \]
 Then we apply Corollary \ref{rt_decomp} to $H_k$ 
 with the intervals $K_{k,c}, K_{k,d}$ such that
 $K_{k,c} \subset I_c, K_{k,d} \subset I_d$ and $K_{1,d} \cap K_{2,c} = \emptyset$
 (see Figure \ref{fig:decomp2}), to obtain operators $C_k, D_k$ such that
 $r_k^{H_k} = C_kD_k$. By the localization, $C_2$ and $D_1$ commute.
 
 Hence it holds that  $r^{L_0} = r_1^{H_1}r_2^{H_2}U(g) = C_1 D_1 C_2 D_2 U(g) = C_1 C_2 D_1 D_2 U(g)$,
 and $C:=C_1C_2$ is localized in $K_{1,c}\cup K_{2,c} \subset I_c$, while
 $D:=D_1D_2$ is localized in $K_{1,d}\cup K_{2,d} \subset I_d$, as desired.


\end{proof}

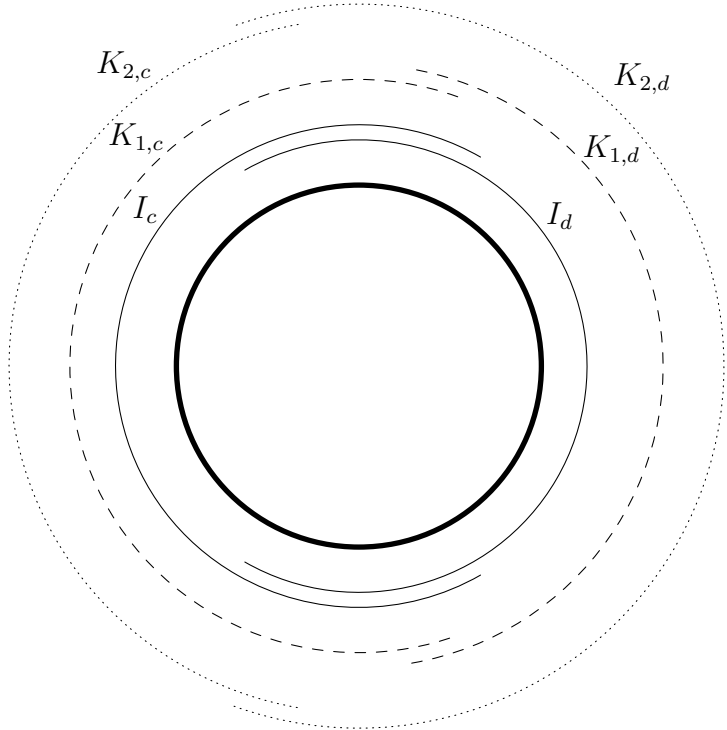
\begin{figure}[ht]
\centering
\begin{tikzpicture}[line cap=round,line join=round,>=triangle 45,x=1.0cm,y=1.0cm, scale= 0.8]
\clip(-10,-6.2) rectangle (10,6.2);
\draw [line width=2pt] (0,0) circle (3cm);
\draw (-1.875,-3.247595264) arc (-120:120:3.75) ;
\draw (2,3.464101615) arc (60:300:4) ;
\draw [dash pattern=on 4pt off 4pt] (1.624595681,4.463539949) arc (70:290:4.75) ;
\draw [dash pattern=on 4pt off 4pt] (0.868240888,-4.924038765) arc (-80:80:5) ;
\draw [dotted] (-0.998477022,5.66264458) arc (100:260:5.75) ;
\draw [dotted] (-2.05212086,-5.638155725) arc (-110:110:6) ;
 \draw (2.9,2.9) node[anchor=north west] {$I_d$};
 \draw (-3.9,3) node[anchor=north west] {$I_c$};
 \draw (-4.3,4.2) node[anchor=north west] {$K_{1,c}$};
 \draw (3.5,4) node[anchor=north west] {$K_{1,d}$};
 \draw (-4.5,5.4) node[anchor=north west] {$K_{2,c}$};
 \draw (4,5.2) node[anchor=north west] {$K_{2,d}$};
\end{tikzpicture}\caption{Intervals $I_c, I_d, K_{1,c}, K_{1,d}, K_{2,c}, K_{2,d}$.}
\label{fig:decomp2}
\end{figure}


\section{Normality of the product vacuum state}\label{normality}

We can now prove our main claim: for a conformal net on $S^1$
--- where by ``conformal'' we mean that it has the full diffeomorphism covariance (see Section \ref{preliminaries}) ---
the split property is automatic. Let $(\a,U,\Omega)$ be a conformal net, and assume $I_a,I_b\in \I$ are two open proper arcs separated by a positive distance. 

Consider the $*$-algebra $\a(I_a)\vee_{\mathrm{alg}}\a(I_b)$ generated by the commuting factors $\a(I_a)$ and $\a(I_b)$.
We shall now introduce a family $\{\phi_z\}$ of functionals on this algebra 
indexed by a complex number $z$, $|z| \le 1$. For a generic element $X\in \a(I_a)\vee_{\mathrm{alg}}\a(I_b)$, 
\begin{equation}\label{genericX}
X= \sum_{k=1}^nA_kB_k\qquad (n\in\NN, \; A_k\in \a(I_a),\; B_k \in \a(I_b))
\end{equation}
and a complex number $z$ in the closed unit disk $\d_1=\{z\in\CC: |z|<1\}$,
let
\begin{equation*}
\phi_z(X)=\sum_{k=1}^n\langle\Omega,A_k z^{L_0}B_k\Omega\rangle.
\end{equation*}
The above quantity is well-defined in the sense that it indeed depends only on $z$ and $X$, but not on the particular decomposition chosen for $X$.
Indeed, since $\a(I_a)$ and $\a(I_b)$ are commuting factors, there is a natural isomorphism between the algebraic tensor product
$\a(I_a)\odot\a(I_b)$ and $\a(I_a)\vee_{\mathrm{alg}}\a(I_b)$, see \cite[Proposition IV.4.20]{Takesaki}.
In particular, the bilinear form $\a(I_a)\times \a(I_b)\ni 
(A,B)\mapsto \langle\Omega,A z^{L_0}B\Omega\rangle \in \CC$ extends to a unique linear functional $\phi_z$ on $\a(I_a)\vee_{\mathrm{alg}}\a(I_b)$. 

Note that the expression $z^{L_0}$ is indeed a well-defined bounded operator for every $z\in \d_1$
(for $z=0$, we define it by continuity in the strong operator topology, hence to be the projection $P_0$ onto $\CC\Omega$):
this is because ${\rm Sp}(L_0)\subset \NN$. That is, we are using not just the positivity of $L_0$, but also that elements of its spectrum are all integers (e.g.\! $z^{\frac{1}{2}} = \sqrt{z}$ would be ambiguous). 

For every $X\in\a(I_a)\vee_{\mathrm{alg}}\a(I_b)$, the map $z\mapsto\phi_z(X)$ is analytic in $\d_1$. In fact, denoting by $P_m$ the spectral projection of $L_0$ associated to the eigenvalue $m$, we have the power series decomposition of $\phi_z(X)$
\begin{equation*}
\phi_z(X)=\sum_{k=1}^n\sum_{m=0}^\infty\langle\Omega,A_kP_mB_k\Omega\rangle z^m.
\end{equation*}
Since $P_0=\langle\Omega,\,\cdot\,\rangle\,\Omega$ is the one-dimensional  projection on the vacuum vector, we have that 
\begin{equation*}
\phi_0(X)=\sum_{k=1}^n \omega(A_k)\omega(B_k),
\end{equation*}
i.e.\! $\phi_0$ is the product vacuum state, whereas $\phi_1=\omega$. Thus, in view of    
Proposition \ref{prop:split}, in order to prove the split property, we need to show that while ``changing'' the parameter $z$ from $1$ to $0$, the functional $\phi_z$ remains normal. In particular, it would be desirable to obtain estimates on $\|\phi_1-\phi_0\|$.

The idea of considering $\phi_z$ not only at the points $z=1$ and $z=0$, but on a larger area
(so that its analytic dependence on $z$ can be exploited) comes from \cite{FJ}. There the authors
work with the function $z\mapsto \phi_z(AB)$ to obtain a bound on $|\phi_1(AB)-\phi_0(AB)|$
for a pair of elements $A\in \a(I_a)$ and $B\in \a(I_b)$
thereby proving the conformal cluster theorem for a M\"obius covariant net. However, 
their estimate involves the product of norms $\|A\|\, \|B\|$; when it is reformulated 
for an element $X$ of the considered form \eqref{genericX}, we get some bounds in terms of
$\sum_{k}\|A_k\|\, \|B_k\|$, rather than in terms of the norm of $X$. Hence their 
method does not give a useful estimate on $\|\phi_1-\phi_0\|$.
In fact, they {\it cannot} obtain anything that would imply the split property:
this is because they only use M\"obius covariance, and as was mentioned in the introduction,
counterexamples to the split property exist when diffeomorphism covariance is not assumed \cite[Section 6]{CW}. 

Instead, our idea is the following: using diffeomorphism covariance and in particular
the decompositions of $r^{L_0}$ established in the previous sections, we can show that $\phi_z$
depends norm-continuously on $z$ and is normal (i.e.\! extends to a normal linear functional of the von Neumann algebra $\a(I_a)\vee\a(I_b)$)
when $z$ is in a certain region. Unfortunately, the region directly obtainable by such decompositions do not contain the desired point $z=0$.
However, if this region contains a ring encircling the point $z=0$ (and as we shall see, this will exactly be the case)
we can use general complex analytic arguments (essentially the Cauchy theorem) to deduce normality of $\phi_0$:

\begin{lemma}
\label{cauchy}
Let $r_0\in (0,1)$ be a fixed radius and suppose that $\phi_z$ is normal whenever
$|z|=r_0$ and that on the circle with radius $r_0$, $r_0S^1\ni z \mapsto \phi_z$ is norm-continuous. Then $\phi_0$ is
also normal.
\end{lemma}
\begin{proof}
We shall use some well-known technical facts.
In particular, we shall exploit that the norm-limit of a sequence of normal functionals on a von Neumann algebra $\M$ is normal
(see e.g.\! \cite[Corollary 7.1.13]{KR}).
To apply this fact, one should note that the norm is defined on the von Neumann algebra $\M$, but
 by the Kaplansky density theorem, the norm of a normal functional on $\a(I_a)\vee \a(I_b)$ is equal
 to the norm of its restriction to $\a(I_a)\vee_{\mathrm{alg}}\a(I_b)$.
 Therefore, in the following we do not distinguish them.

Thus one has --- e.g.\! by considering Riemann-sums ---
that if $\varphi:[s_1,s_2]\ni t\mapsto \varphi_t\in \M_*$ is a norm-continuous family of normal linear functionals,
then $\varphi(\cdot)=\int_{s_1}^{s_2}\varphi_t(\cdot) dt$ is also a well-defined normal functional on $\M$.

Since $\d_1\ni z\mapsto \phi_z(X)$ is analytic, by the Cauchy integral formula we have
\begin{equation*}
\phi_0(X) = \frac{1}{2\pi i}\oint_{r_0S^1} \phi_{r_0 e^{i\theta}}(X)\; \frac{dz}z
\end{equation*}
for every $X\in \a(I_a)\vee_{\mathrm{alg}}\a(I_b)$.

\end{proof}

Let us now discuss how the decompositions $r^{L_0}$ help us out in different regions of $\d_1$. Let $I_c=I_a'$, and $I_d$ be an (open) interval containing the closure of $I_a$ but not intersecting $I_b$. (Such an ``enlargement'' of $I_b$ exists as $I_a$ and $I_b$ were assumed to have a positive distance from each other). We then have that $I_c\cup I_d=S^1$ and we can consider the decomposition of $r^{L_0}$ given by Proposition \ref{prop:dec1} with  $\{C_r\}_{r\in(0,1)}\subset \a(I_c)$ and $\{D_r\}_{r\in(0,1)}\subset \a(I_d)$.
Let us denote by $R_\theta$ the rotation by $\theta$.
Then, as long as $z=r e^{i\theta}$ is such that $r\in (0,1)$ and the angle $\theta$ satisfies the condition
\begin{equation}\label{eq:theta1set}
\overline{I_d}\cap R_\theta (\overline{I_b}) =\emptyset,
\end{equation}
we have that the action of $\rho_\theta \equiv \Ad e^{i\theta L_0}$ leaves $\a(I_b)$ inside $\a(I_d)'$ and thus 
for an $X \in\a(I_a)\vee_{\mathrm{alg}} \a(I_b)$ 
with decomposition (\ref{genericX}), we can use {\it locality} to rewrite $\phi_z(X)$  as
\begin{eqnarray}
\nonumber
\phi_z(X)&=&\sum_{k=1}^n\left\langle\Omega,A_k r^{L_0} e^{i\theta L_0} B_k\Omega\right\rangle
=\sum_{k=1}^n\left\langle\Omega,A_k r^{L_0} e^{i\theta L_0} B_k e^{-i\theta L_0}\Omega\right\rangle
\\
\label{eq:phistep1}
&=&\sum_{k=1}^n\left\langle\Omega,A_k C_rD_r \rho_\theta(B_k)\Omega\right\rangle
=\left\langle C_r^*\Omega, \left(\sum_{k=1}^n A_k \rho_\theta(B_k)\right)D_r\Omega\right\rangle.
\end{eqnarray}
Let $(\theta_-,\theta_+)$ be the largest open interval of angles satisfying our condition \eqref{eq:theta1set} and containing $0$,
i.e., $\theta_\pm$ is the smallest positive / largest negative angle for which $R_{\theta_\pm}(\overline{I_b})$ intersects $\overline{I_d}$.
We can obviously find a smooth function $f:S^1\to \RR$ such that $f$ on $I_a$ is zero, but is constant $1$ on the complement of $I_d$ (i.e.\! on the complement of the ``enlarged'' version of $I_a$).
Viewing $f$ as the vector field on $S^1$ formally written as $f(e^{i\theta})\frac{\rm d}{{\rm d}\theta}$, it gives rise to a one parameter group of diffeomorphisms  
\begin{equation*}
\RR\ni\theta\mapsto \gamma_\theta \equiv \,\exp(\theta f)
\end{equation*}
such that $\gamma_\theta$ is ``localized'' in $I_a'$, but if $\theta_-<\theta<\theta_+$, then the action of $\gamma_\theta$ on $I_b$
coincides with that of the rotation by $\theta$. Thus, $e^{i\theta T(f)}$ commutes with elements of $\a(I_a)$
but for $\theta \in (\theta_-,\theta_+)$, its adjoint action on $\a(I_b)$ coincides with the action of $\rho_\theta$ and hence 
we can write
\begin{eqnarray*}\nonumber
\sum_{k=1}^nA_k\rho_\theta (B_k)&=&
\sum_{k=1}^nA_k e^{i\theta T(f)}B_k e^{-i\theta T(f)} = 
e^{i\theta T(f)} \sum_{k=1}^nA_k B_k e^{-i\theta T(f)} = 
\\
\label{eq:sums}
&=&e^{i\theta T(f)} X e^{-i\theta T(f)}.
\end{eqnarray*} 
Putting this back in \eqref{eq:phistep1}, we get that for $\theta \in (\theta_-,\theta_+)$, 
\begin{equation*}\label{eq:decdec1}
\phi_{re^{i\theta}}(X)=\langle \eta_{\theta},X \zeta_{\theta}\rangle
\end{equation*} 
where the vectors $\eta_{\theta} = e^{-i\theta T(f)}C_r^*\Omega$
and $\zeta_{\theta} =  e^{-i\theta T(f)}D_r\Omega$.
\begin{corollary}\label{co:cont1}
$\{\phi_z\}$ is a norm-continuous family of  normal functionals in the region $\{re^{i\theta}| r\in(0,1)\;\;{\rm  and  }\;\; \overline{I_d}\cap R_\theta(\overline{I_b}) = \emptyset \}.$  
\end{corollary}

Note that Proposition \ref{prop:dec1} gives some bounds on the norms of $C_r$ and $D_r$, and so actually with the constant $q>0$ defined there, in the discussed region
we have
\begin{equation*}
\|\phi_{re^{i\theta}}\|\leq \| C_r^*\Omega\|\; \|D_r\Omega\|
\leq\frac1{r^{2q}}.
\end{equation*}
Unfortunately, though this estimate is nicely uniform in $\theta$, it ``blows up'' at $r\to 0$ and hence in itself it does not show that $\phi_r$ converges to $\phi_0$ in norm as $r\to 0$. 

However, so far we have only used our first decomposition of $r^{L_0}$.
We shall now exploit the second one derived in Section 4.
Let us now consider the decomposition given by Proposition \ref{prop:dec2} for a certain (fixed) $r_0 \in (0,1)$;
that is, we have two bounded elements $C\in \a(I_c)$ and $D\in \a(I_d)$ and $g\neq {\rm id}$
a M\"obius transformation such that $r_0^{L_0} = CD U(g)$ (recall that this is valid in the proper sense,
namely one can fix the phase of $U(g)$ for $g \in \mob$ unambiguously).
Then, repeating the steps we did before with our previous decomposition and setting instead of \eqref{eq:phistep1}, this time we get
\begin{equation*}
\phi_{r_0e^{i\theta}}(X) = \sum_{k=1}^n\left\langle \Omega, A_k CD U(g)e^{i\theta L_0}B_k\Omega\right\rangle
=\left\langle C^*\Omega,\left(\sum_{k=1}^n A_k (\Ad U(g)\circ\rho_\theta)(B_k) \right)D\Omega\right\rangle
\end{equation*}
whenever the disjointness condition 
\begin{equation*}
\overline{I_d}\cap g \circ R_\theta (\overline{I_b}) =\emptyset
\end{equation*} 
holds.
We claim that there is a continuous family of diffeomorphisms parametrized by $\theta$ which coincide with
$g \circ R_\theta$ on $I_b$ and are trivial on $I_d$ (this time we should take diffeomorphisms not necessarily of the form
$\Exp(\theta f)$): indeed, we can identify these diffeomorphisms with a common fixed point
with smooth functions on $\RR/2\pi\ZZ$ onto $\RR/2\pi\ZZ$ with strictly positive derivatives,
and then with strictly positive functions on $\RR/2\pi\ZZ$ with integral $2\pi$.
Then for given derivatives on small neighborhoods of $\overline{I_b}$ and $\overline{I_d}$, we can fill the complement
of $I_b \cup I_d$ by a smooth positive function with the correct integrals (so that the left boundary point of $I_b$
is mapped to the correct point).
Note that the last step is possible since $(g\circ R_\theta)I_b\cap  I_d = \emptyset$
and the lenght of $(g\circ R_\theta)I_b\cup  I_d$ is less then $2\pi$.
This gives a required continuous family.
Now we can continue exactly as in the first case, and hence this time obtain the following.
\begin{corollary}\label{co:cont2}
$\{\phi_z\}$ is a norm-continuous family of  normal functionals in the region $\{r_0 e^{i\theta}| \overline{I_d}\cap g\circ R_\theta(\overline{I_b}) = \emptyset \}.$  
\end{corollary}

Does the union of the two treated regions encircle the point $0$? This might not be the case. However, note that the M\"obius transformation $g$ given by Proposition \ref{prop:dec2} is an ``absolute'' one; i.e.\! it does not depend on the intervals $I_c$ and $I_d$ (whereas of course the elements $C$
and $D$ obviously do). And though for some choices of $I_a, I_b$ and $I_d\supset \overline{I_a}$ might lead to nowhere, it is enough for us to show that {\it there is} a ``right'' choice.

\begin{theorem}
A conformal net $(\a,U,\Omega)$ on the circle automatically has the split property.
\end{theorem}
\begin{proof}
By conformal covariance, we may assume that $I_a, I_b$ and even the enlargement $I_d\supset \overline{I_a}$
are ``tiny''; almost point-like intervals around two points which we will conveniently call $a$ and $b$.
Then the region guaranteed by Corollary \ref{co:cont1} is $\d_1$ minus a slightly enlarged version of the half-line $\{te^{i\alpha}| t\geq 0\}$
where $\alpha$ is the  angle for which $R_\alpha(b)=a$.
On the other hand, the region guaranteed by Corollary \ref{co:cont2} is the circle $r_0 S^1$ minus a slightly enlarged version of
the point $r_0 e^{i\tilde{\alpha}}$, where $\tilde{\alpha}$ is the angle for which $g\circ R_{\tilde{\alpha}}(b)=a$,
which is of course equivalent to saying that $R_{\tilde{\alpha}}(b)=g^{-1}(a)$.
Since $g$ is a certain fixed, non trivial M\"obius transformation, we might even assume that our choice of $a$ is such that $g^{-1}(a)\neq a$. Then $\alpha\neq \tilde{\alpha}$ and the union of the two regions covers the circle $r_0 S^1$.

Lemma \ref{cauchy} shows that $\phi_0$ is a normal linear functional on $\a(I_a)\vee \a(I_b)$.
Now our claim is concluded by Proposition \ref{prop:split} and a technical Lemma \ref{lm:state} below,
by noting that
\begin{itemize}
 \item $\a(I_a)\vee \a(I_b)$ is a factor (Section \ref{preliminaries}, Factoriality of two-interval algebras)
 \item The restrictions of $\phi_0$ to $\a(I_a)$ and $\a(I_b)$ are equal to the vacuum state,
 hence faithful.
\end{itemize}
\end{proof}

 
\begin{lemma}\label{lm:state}
 $\phi_0$ is a positive normal functional on  $\a(I_a)\vee\a(I_b)$.
\end{lemma}
\begin{proof}
 We first consider $\phi_0$ on $\a(I_a)\vee_{\mathrm{alg}}\a(I_b)$.
 By \cite[Proposition 4.20]{Takesaki},
 the map $\tau:\sum_k x_k y_k \mapsto \sum x_k\otimes y_k$ is well-defined and is a $*$-isomorphism from
 $\a(I_a)\vee_{\mathrm{alg}}\a(I_b)$ onto $\a(I_a)\odot\a(I_b)$.
 Now the linear functional $\phi_0(\sum_k x_k y_k)$ translates into $\a(I_a)\odot\a(I_b)$
 as $\<\Omega\otimes\Omega, \cdot\;\Omega\otimes\Omega\>$.
 Namely, $\phi_0 = (\omega\otimes\omega) \circ \tau^{-1}$.
  Now, $\omega\otimes\omega$ is clearly positive, and $\tau^{-1}(x^*x) = \tau^{-1}(x)^*\tau^{-1}(x)$,
 therefore, $\phi_0(x^*x) \ge 0$ for any $x \in \a(I_a)\vee_{\mathrm{alg}}\a(I_b)$.
 
 We claim that, also on $\a(I_a)\vee_{\mathrm{alg}}\a(I_b)$, $\phi_0$ is positive
 \footnote{$\a(I_a)\vee_{\mathrm{alg}}\a(I_b)$ is not a $C^*$-algebra, in particular,
   a positive element $a \in \a(I_a)\vee_{\mathrm{alg}}\a(I_b)$ in the sense of $\b(\h)$ is not necessarily of the form $x^*x$, where $x\in \a(I_a)\vee_{\mathrm{alg}}\a(I_b)$.}.
 Indeed, take a positive element $a \in \a(I_a)\vee_{\mathrm{alg}}\a(I_b)$.
 The function $f(x) = x^\frac12$, $x \in \left[0, \|a\|\right]$ can be arbitrarily approximated by polynomials $f_n$
 with real coefficients, uniformly on $\left[0, \|a\|\right]$ and $f_n(a)^2$ tends to $a$ in norm.
 We saw that $\phi_0$ is a normal linear functionals, hence it is in particular
 continuous in norm.
 Since $\phi_0(x^*x)\geq 0$ for $x\in \a(I_a)\vee_{\mathrm{alg}}\a(I_b)$ then $\phi_0(f_n(a)^2) \ge 0$, hence $\phi_0(a) \ge 0$ by norm continuity of $\phi_0$.

 Now, by the Kaplansky density theorem and the normality of $\phi_0$, $\phi_0$ is a positive functional.

%
\end{proof}

\section{A non-split conformal net in two-dimensions}\label{non-split}
%
Conformal nets on $S^1$ constitute the building blocks of two-dimensional conformal nets.
Let us recall the relevant definitions.
A locally normal, positive energy, M\"obius covariant \textbf{representation}  $\rho$  of a conformal net $(\a,U,\Omega)$ on $S^1$
is a family of normal representations $\{\rho_I:I\in\I\}$  of the von Neumann algebras $\{\a(I): I \in \I\}$ on a fixed Hilbert space
$\h_\rho$ and a  unitary, positive energy unitary representation $U^\rho$ on $\h_\rho$
of the universal covering group of the M\"obius group $\widetilde{\mob}$ satisfying:
\begin{enumerate}
\item Compatibility: if $I_1,I_2\in\I$ and $I_1\subset I_2$ then $\rho_{I_2}|_{\a(I_1)}=\rho_{I_1}$
\item Covariance: $\Ad U^\rho(g)\circ \rho_I=\rho_{gI}\circ \Ad U(g)$, $g\in\widetilde\mob$ 
\end{enumerate}
A representation $\rho$ is irreducible if $\bigvee_{I\in\I}\rho(\a(I))=\b(\h_\rho)$.
The defining representation $\{\mathrm{id}_{\a(I)}\}$ is called the \textbf{vacuum representation}.


A representation of a conformal net $\rho$ is said to be \textbf{localizable}  in $I_0$
if $\rho_{I_0'}\simeq \mathrm{id}$, where $\simeq$ means unitary equivalence. The unitary equivalence class of $\rho$ defines a \textbf{superselection sector}, also called a \textbf{DHR (Doplicher-Haag-Roberts)  sector}
\cite{DHR69}. 
By Haag duality we have that $\rho(\A(I))\subset\A(I)$ if $I_0\subset I$. Thus we can always choose, within the sector of $\rho$,
a representation $\rho_0$ on the defining Hilbert space $\h$ such that
$\rho_{0,I_0}$ is an endomorphism of $\A(I_0)$.
If each $\rho_I$ is an automorphism of $\A(I)$, we call $\rho$ an \textbf{automorphism} of $(\A,U,\Omega)$.
Automorphisms can be composed in a natural way.

Let $(\A, U, \Omega)$ be the $\mathrm{U}(1)$-current net \cite{BMT88}.
The main ingredients are (see \cite{weinerthesis} for a more detailed review):
\begin{itemize}
 \item The Weyl operators $W(f)$ parametrized by real smooth functions $f$ on $S^1$
 which satisfy the commutation relations $W(f)W(g) = e^{\frac i 2(f,g)}W(f+g)$,
 where $(f,g) := \frac1{2\pi}\int_0^{2\pi}d\theta\,f'(e^{i\theta})g(e^{i\theta})$
 and $f'(e^{i\theta}) = \frac d {d\theta}f(e^{i\theta})$.
 \item There is a distinguished realization (``vacuum representation'') of the Weyl operators
 (which we denote again by $W(f)$) with a unitary positive energy
 representation of $\mob$ which extends to a projective unitary representation $U$ of $\diff$, and the vacuum vector $\Omega$ such that
 $\Ad U(\g) (W(f)) = W(f\circ \g)$ and $U(g)\Omega = \Omega$ if $g\in\mob$.
 \item The $\mathrm{U}(1)$-current net $\A(I) := \{W(f): \supp f \subset I\}''$.
 \item Irreducible sectors parametrized by $q \in \RR$:
 we fix a real smooth function $\varphi$ such that $\frac1{2\pi}\int_0^{2\pi}d\theta\,\varphi(e^{i\theta}) = 1$.
 The map $W(f) \to e^{iq\varphi(f)}W(f)$ extends to an automorphism $\s_{q,I}$ of $\A(I)$, where
 $\supp f \subset I$ and $\varphi(f) = \frac1{2\pi}\int_0^{2\pi}d\theta\,f(e^{i\theta})\varphi(e^{i\theta})$.
 We call this automorphism of the net $\s_q$.
 Different functions $\varphi$ with the conditions above with the same $q$ give the equivalent sectors,
 while sectors with different $q$ are inequivalent.
 It holds that $\s_q\circ\s_{q'} = \s_{q+q'}.$
 \item Each irreducible sector is covariant: the projective representation $\g\mapsto U_q(\g) := \s_q(U(\g))$
 of local diffeomorphisms extends to $\widetilde{\diff}$, hence makes the automorphism $\s_q$ covariant
 \cite[Proposition 2]{DFK04} (in an irreducible representation $\s_q$, the choice of $U_q(\g)$ is unique
 up to a scalar \cite[Remark after Proposition 2]{DFK04}):
 $\Ad U_q(\g)(\s_q(x)) = \s_q(\Ad U(\g)(x))$. Furthermore, we can fix the phase of $U_q(\g)$ and
 consider them as unitary operators (see \cite[Proposition 5.1]{FH}, where the phase does not depend on $h$,
 hence one can take the direct sum of multiplier representations (projective representations with fixed phases)).
 In this case, it holds that $U_q(\g_1)U_q(\g_2) = c(\g_1,\g_2)U(\g_1,\g_2)$
 where $c(\g_1,\g_2) \in \CC\1$. $c(\g_1,\g_2)$ can be chosen without dependence on $q$,
 and continuous in a neighborhood of the unit element.
 This projective representation (restricted to $\widetilde{\mob}$) has positive energy \cite{CW}.
 \item For two equivalent automorphisms $\rho, \tilde\rho$ localized in $I, \tilde I$, respectively,
 an operator which intertwines them is called a charge transporter. In the present case, as both $\rho, \tilde\rho$
 are irreducible, such a charge transporter is unique up to a scalar. A charge transporter acts trivially
 on $\A((I \cup \tilde I)')$, hence belongs to $\A((I \cup \tilde I)')'$. In particular, it can be considered
 as an element in a local algebra containing $I$ and $\tilde I$.
 \item The operator $z_q(\g) := U(\g)U_q(\g)^*$ is a charge transporter between
 $\s_q$ and $\alpha_\g\s_q\alpha_{\g^{-1}}$.
 \item  For a given pair of automorphisms $\rho_1, \rho_2$, one defines the braiding $\epsilon_{\rho_1, \rho_2}$:
 one chooses equivalent automorphisms $\tilde\rho_1, \tilde\rho_2$ localized in $\tilde I_1, \tilde I_2$,
 respectively, such that $\tilde I_1 \cap \tilde I_2 = \emptyset$ and charge transporters
 $V_1, V_2$ between $\rho_1$ and $\tilde \rho_1$, and $\rho_2$ and $\tilde \rho_2$,
 respectively. Define $\epsilon_{\rho_1, \rho_2}^\pm := \rho_2(V_1^*)V_2^*V_1\rho_1(V_2)$, where $+$ or $-$ depends on
 the choice whether $\tilde I_1$ is on the left/right of $\tilde I_2$ (which results from the choice of localization of the charge transporter above),
 but $\epsilon_{\rho_1, \rho_2}^\pm$ do not depend on the choice of $\tilde \rho_k, V_k$ under such a configuration.
 \item For our concrete automorphisms $\s_q, \s_{q'}$ on the $\mathrm{U}(1)$-current net,
 one can take the charge transporters $V_q, V_{q'}$ as Weyl operators and finds that the braiding satisfies
 $\epsilon_{\s_q, \s_{q'}}^\pm \in \CC\1$, $\overline{\epsilon_{\s_q, \s_{q'}}^+} = \epsilon_{\s_q, \s_{q'}}^-$.
 
\end{itemize}

The following is probably well known to experts, but it is difficult to find the right reference
(for example, \cite[Proposition 1.4]{Longo97} is proved for M\"obius covariance).
We note that a systematic formulation, closer to our needs, is to appear in \cite{DG16}.
Nevertheless, in part because we deal with multiplier representations,
and in part for better readability, we include a formal statement with a proof.

\begin{proposition}[Tensoriality of cocycles]
 It holds that $z_q(\g)\s_q(z_{q'}(\g)) = z_{q+q'}(\g)$.
\end{proposition}
\begin{proof}
 First recall that $z_q(\g)$ is an intertwiner between $\s_q$ and $\alpha_\g\s_q\alpha_{\g^{-1}}$,
 hence the product $z_q(\g)\s_q(z_{q'}(\g))$ is an intertwiner between
 $\s_q\s_{q'} = \s_{q+q'}$ and $\alpha_\g\s_q\alpha_{\g^{-1}}\circ \alpha_\g\s_{q'}\alpha_{\g^{-1}} = \alpha_\g\s_{q+q'}\alpha_{\g^{-1}}$.
 $z_{q+q'}(\g)$ also intertwines  $\s_{q+q'}$ and $\alpha_\g\s_{q+q'}\alpha_{\g^{-1}}$.
 As they are automorphisms, hence irreducible, the difference between $z_q(\g)\s_q(z_{q'}(\g))$ and $z_{q+q'}(\g)$
 must be a scalar.
 
 Next we show that $U'_{q+q'}(\g) := (z_q(\g)\s_q(z_{q'}(\g)))^*U(\g)$
 is a multiplier representation of $\widetilde{\diff}$ such that
 $U_{q+q'}'(\g_1)U_{q+q'}'(\g_2) = c(\g_1,\g_2)U_{q+q'}'(\g_1\g_2)$, namely
 it has the same $2$-cocycle $c$ as $U_{q+q'}$.
 Indeed,
 \begin{align*}
  U'_{q+q'}(\g_1)U'_{q+q'}(\g_2) &= (z_q(\g_1)\s_q(z_{q'}(\g_1)))^*U(\g_1) (z_q(\g_2)\s_q(z_{q'}(\g_2)))^*U(\g_2) \\
 &= \s_q(z_{q'}(\g_1))^*U_q(\g_1) \s_q(z_{q'}(\g_2))^*U_q(\g_2) \\
 &= \s_q(z_{q'}(\g_1))^* \s_q(\alpha_{\g_1}(z_{q'}(\g_2)))^*\cdot c(\g_1,\g_2)U_q(\g_1\g_2) \\
 &= \s_q(z_{q'}(\g_1)^* \alpha_{\g_1}(z_{q'}(\g_2))^*)\cdot c(\g_1,\g_2)U_q(\g_1\g_2) \\
 &= \s_q\left(U_{q'}(\g_1)U(\g_1)^* U(\g_1)U_{q'}(\g_2)U(\g_2)^*U(\g_1)^*\right)\cdot c(\g_1,\g_2)U_q(\g_1\g_2) \\
 &= \s_q(U_{q'}(\g_1\g_2) U(\g_1\g_2)^*)\cdot c(\g_1,\g_2)z_q(\g_1\g_2)^*U(\g_1\g_2) \\
 &= c(\g_1,\g_2)U_{q+q'}'(\g_1\g_2),
 \end{align*}
 where in the 3rd and 6th equalities we used that $U$ and $U_q$ share the same $2$-cocycle $c$.
 
 Now let us define $U''(\g) := U_q'(\g)^*U_q(\g)$.
 As the difference between $U_q'(\g)$ and $U_q(\g)$ is just a phase and they share the same $2$-cocycle $c$,
 it is easy to show that $U''$ is a $\CC$-valued true (with trivial multiplier) representation of $\widetilde{\diff}$.
 It is well-known that then $U''$ must be trivial, $U''(\g) = \1$. From this the claim immediately follows.
\end{proof}

Let $G$ be the quotient of $\widetilde\mob\times\widetilde\mob$ by the normal subgroup generated by
$(R_{2\pi}, R_{-2\pi})$, where $\widetilde\mob$ naturally includes the universal covering $\RR$ of the rotation
subgroup $S^1$ and $R_{2\pi}, R_{-2\pi}$ are the elements corresponding to $2\pi, -2\pi$ rotations, respectively.
We call $\RR \times S^1$ the Einstein cylinder $\cyl$,
where the Minkowski space is identified with a maximal square $(-\pi,\pi)\times (-\pi,\pi)$ (see \cite{BGL93})
\footnote{Here the segments $(-\pi, \pi)\times \{0\}$ and $\{0\}\times (-\pi,\pi)$ are identified with
the time and space axis, respectively.}.
The group $G$ acts naturally on it. Furthermore, 
let $\mathrm{Diff}(\RR)$ be the group of diffeomorphisms of $S^1$ which preserves the point of infinity $\infty$,
with the identification $S^1 = \RR \cup \{\infty\}$.
Then $\mathrm{Diff}(\RR)\times \mathrm{Diff}(\RR)$ acts naturally
on the Minkowski space as the product of two lightrays
\footnote{The lightray decomposition $\RR^2 = \RR\times \RR$ is not compatible with the above
identification of $\RR$ with $(-\pi,\pi)\times (-\pi,\pi)$, where the components correspond
to the time and space axis.}, and its action naturally extends to $\cyl$ by periodicity.
Let us denote by $\conf$ the group generated by $G$ and $\mathrm{Diff}(\RR)\times \mathrm{Diff}(\RR)$.
A two-dimensional conformal net $(\tilde \a, \tilde U, \tilde \Omega)$
consists of a family $\{\tilde \a(O)\}$ of von Neumann algebras parametrized by double cones $\{O\}$ in the Minkowski space $\RR^2$,
a strongly-continuous unitary representation of $G$ which extends to a projective unitary
representation of $\conf$, and a vector $\tilde \Omega$ such that
the following axioms are satisfied \cite[Section 2]{KL04-2}:
\begin{itemize}
 \item {\sc Isotony}. If $O_1 \subset O_2$, then $\tilde \a(O_1)\subset \tilde\a(O_2)$.
 \item {\sc Locality}. If $O_1$ and $O_2$ are spacelike separated, then $\tilde\a(O_1)$ and $\tilde\a(O_2)$ commute.
 \item {\sc Covariance}. For a double cone $O$, it holds that $\Ad \tilde U(\g)(\tilde \a(O)) = \tilde \a(\g O)$
 for $\g \in \mathcal{V} \subset \conf$, where $\mathcal{V}$ is a neighborhood of the unit element of $\conf$
 such that $\g O \subset \RR^2$ for $\g\in\mathcal{V}$.
 For $x \in \tilde\a(O)$ and if $\g \in \mathrm{Diff}(\RR)\times\mathrm{Diff}(\RR)$ acts identically on $O$, then $\Ad \tilde U(\g)(x) = x$.
 \item {\sc Existence and uniqueness of vacuum}. $\tilde\Omega$ is a unique (up to a scalar) invariant vector
 for $\tilde U|_G$.
 \item {\sc Cyclicity}. $\tilde\Omega$ is cyclic for $\bigvee_{O \subset \RR^2} \tilde\a(O)$.
 \item {\sc Positivity of energy}. The restriction of $\tilde U$ to the group of translations has the spectrum
 contained in $V_+ := \{(x_0,x_1): x_0 \ge |x_1|\}$.
\end{itemize}

Now we construct a two-dimensional conformal net as follows, following the ideas of \cite{DR90, LR95}.
Let us fix an interval $I \subset \RR \subset S^1$ and a real smooth function $\varphi$
as above. On the Hilbert space $\h_q = \h$, we take the automorphism $\s_q$ of the
$\mathrm{U}(1)$-current net $\A$.
The full Hilbert space is the separable direct sum $\tilde\h = \bigoplus_{q\in\QQ} \h_q\otimes\h_q$.
The observable net $\A\otimes\A$ acts on $\tilde\h$ as the direct sum $\tilde\s(x\otimes y) = \bigoplus_q \s_q(x)\otimes\s_q(y)$.
We can also define a multiplier representation of $\widetilde{\diff}\times\widetilde{\diff}$ by
$\tilde U(\g_+,\g_-) := \bigoplus_q U_q(\g_+)\otimes U_q(\g_-)$.
The representation $\tilde U$ actually factors through $\conf$.
This can be seen by noting that in each component $U_q\otimes U_q$ the generator of spacelike
rotations is $L_0^{\s_q}\otimes\1-\1\otimes L_0^{\s_q}$ whose spectrum is included in $\ZZ$,
since the spectrum of $L_0^{\s_q}$ is included in $\NN + \frac{q^2}2$.

As all the components are the same $\h_q\otimes\h_q = \h\otimes\h$, the shift operators $\{\psi^q\}$ (``fields'') act
naturally on $\tilde\h$: for $\Psi \in \tilde\h$, where $(\Psi)_q \in \h_q\otimes\h_q$, 
\begin{equation*}
 (\psi^{q'}\Psi)_q = (\Psi)_{q+q'}.
\end{equation*}
It is useful to note how they behave under covariance:
\begin{align*}
 (\Ad \tilde U(\g_+,\g_-)(\psi^{q'})\Psi)_q &= U_q(\g_+)\otimes U_q(\g_-)(\psi^{q'}\cdot \tilde U(\g_+, \g_-)^*\Psi)_q \\
 &= U_q(\g_+)\otimes U_q(\g_-)(\tilde U(\g_+, \g_-)^*\Psi)_{q+q'} \\
 &= \left(U_q(\g_+)\otimes U_q(\g_-)\right)\cdot \left(U_{q+q'}(\g_+)^*\otimes U_{q+q'}(\g_-)^* \right)(\Psi)_{q+q'} \\
 &= (z_q(\g_+)^*z_{q+q'}(\g_+))\otimes (z_q(\g_-)^*z_{q+q'}(\g_-))(\Psi)_{q+q'} \\
 &= \left(\s_q(z_{q'}(\g_+)))\otimes (\s_q(z_{q'}(\g_-))\right)(\Psi)_{q+q'} \\
 &= (\tilde\s(z_{q'}(\g_+)\otimes z_{q'}(\g_-))\psi^{q'}\Psi)_q 
\end{align*}
where we used tensoriality of cocycles in the 5th equality.

We define the local algebra, first for $I\times I \subset \RR \times \RR \subset \RR^2$,
where the real lines are identified with the lightrays $x_0 \pm x_1 = 0$, by
\[
 \tilde\A(I\times I) = \{\tilde\s(x\otimes y), \psi^q: x,y\in\A(I), q\in \QQ\}'',
\]
and for other bounded regions by covariance: take $\g_\pm \in \mathrm{Diff}(\RR)$ such that $\g_\pm I = I_\pm$ and
\[
 \tilde\A(I_+\times I_-) = \Ad \tilde U(\g_+, \g_-)(\A(I\times I)).
\]
This does not depend on the choice of $\g_\pm$. Indeed, if $\g_\pm$ preserves $I$, then
$z_{q'}(\g_+)\otimes z_{q'}(\g_-) \in \A(I)\otimes \A(I)$ and 
$\Ad \tilde U(\g_+, \g_-)(\psi^{q'}) \in \tilde\A(I\times I)$ by above computation.
We set $\tilde \Omega = \Omega \otimes\Omega \in \h_0\otimes\h_0 \subset \tilde \h$.

\begin{itemize}
 \item Covariance. $\Ad \tilde U(\g_+,\g_-)(\tilde \a(O)) = \tilde \a((\g_+,\g_-)\cdot O)$ holds by definition.
 If $(\g_+,\g_-) \in \mathrm{Diff}(\RR)\times\mathrm{Diff}(\RR)$ acts trivially on $I\times I$,
 then $\tilde U(\g_+,\g_-) = \tilde\s(U(\g_+)\otimes U(\g_-))$ and this commutes with $\tilde\a(I\times I)$,
 as $\supp \g_\pm$ are disjoint from $I$.
 \item Isotony. By covariance, we may assume that $I_\pm \supset I$.
 Take $\g_\pm$ such that $\g_\pm I = I_\pm$. From the expression
 \[
 \Ad \tilde U(\g_+, \g_-)(\psi^{q'}) =  (\tilde\s(z_{q'}(\g_+)\otimes z_{q'}(\g_-))\psi^{q'}\Psi)_q 
 \]
 and from the fact that $z_{q'}(\g_\pm) \in \A(I_\pm)$, the isotony follows.
 \item Positivity of energy. Each component $U_q\otimes U_q$ has positive energy.
 \item Existence and uniqueness of the vacuum. Only $U_0\otimes U_0$ contains the vacuum vector.
 \item Cyclicity. The fields $\psi^q$ brings $\h_0\otimes \h_0$ to any $\h_q \otimes \h_q$, while
 the local algebra $\tilde\s(\A(I)\otimes\A(I))$ acts irreducibly on each $\h_q \otimes \h_q$. 
 \item Locality. In the two-dimensional situation, the spacelike separation of $I\times I$ and $I_+\times I_-$
 means either $I_+$ sits on the left of $I$ and $I_-$ on the right, or vice versa. We may assume the former case,
 as the latter is parallel.

 The commutativity between the observables $\tilde\s(x\otimes y)$ is trivial.
 As for the observables and the fields $\{\psi^q\}$, if $x, y \in \A(I_\pm)$ respectively, as $I_\pm$ are disjoint from $I$ and $\s_q$ are localized in $I$, we have
 $\tilde\s(x\otimes y) = \bigoplus_q x\otimes y$ and this commutes with shifts $\psi^q$.
 Finally, we need to check the commutativity between fields $\psi^{q_1}, \Ad U(\g_+)\otimes U(\g_-)(\psi^{q_2})$,
 where $\g_\pm I = I_\pm$.
 We can compute the commutator explicitly:
 \begin{align*}
  & ([\psi^{q_1}, (\Ad \tilde U(\g_+)\otimes \tilde U(\g_-)(\psi^{q_2})]\Psi)_q \\
 &= (\psi^{q_1} \tilde\s(z_{q_2}(\g_+)\otimes z_{q_2}(\g_-))\psi^{q_2}\Psi - \tilde\s(z_{q_2}(\g_+)\otimes z_{q_2}(\g_-))\psi^{q_2} \psi^{q_1}\Psi)_q \\  
 &= (\tilde\s(\s_{q_1}(z_{q_2}(\g_+))\otimes \s_{q_1}(z_{q_2}(\g_-)))\psi^{q_1+q_2}\Psi - \tilde\s(z_{q_2}(\g_+)\otimes z_{q_2}(\g_-))\psi^{q_1+q_2}\Psi)_q, 
 \end{align*}
 and this vanishes because $z_{q_2}(\g_+)^*\s_{q_1}(z_{q_2}(\g_+))\otimes z_{q_2}(\g_-)^*\s_{q_1}(z_{q_2}(\g_-))
 = \epsilon_{q_1,q_2}^+\otimes \epsilon_{q_1,q_2}^- = \1$,
 as the braidings $\epsilon_{q_1,q_2}^\pm$ are scalar and conjugate to each other. 
\end{itemize}


We are going to show that $\tilde \a$ does not satisfy the split property. 
First of all by construction, the net $\tilde \a$ satisfies Bisognano and Wichmann property.
Let $q\in\mathbb Q$, $W=\RR^+\times\RR^-$ be a wedge region, by the identification of the Connes-Radon-Nykodym cocycle
with the geometric cocycle \cite[Theorem 2.4]{Longo97}, we get that
\[
 \left.\Delta_{\tilde\a(W),\tilde\Omega}^{-it}\right|_{\h_q\otimes\h_q} = U^q(\Lambda_W(2\pi t))=\Delta_{\a(W),\xi_q,\Omega}^{-it}
\]
where $\xi_q$ is the vector in $\h_q\otimes\h_q$
representing $\phi_q(\cdot)=\omega\circ\sigma^{-1}_q$ on $\a(W) $ and $\Delta_{\a(W),\xi_q,\Omega}$ is the positive part in the polar decomposition of closure of the {\bf relative Tomita operator}
\[
 S_{\a(W),\xi_q,\Omega}:\a(W)\Omega\ni a\Omega\longmapsto a^*\xi_q\in\a(W)\xi_q.
\]
More on relative modular Tomita operators can be found in \cite{BR1}.
Note that $\phi_q \to \omega$ in norm as $q \to 0$,
since $\s_q$ are locally implemented by $W(q\varphi_1)$ for some smooth function $\varphi_1$
which tend to $\1$ strongly as $q\to 0$,
and $\phi_q = \<W(q\varphi_1)^*\Omega, \cdot\, W(q\varphi_1)^*\Omega\>$.

A necessary condition for the split property of the net $\tilde\a$ is the compactness of the completely positive map
\[
  X_0:\tilde\a(O)\ni \tilde a\mapsto\Delta^{\frac14}_{\tilde\a(W),\tilde\Omega}\tilde a\tilde\Omega\in\tilde \h 
\]
where $O\Subset W$ is a double cone \cite[Propositions 1.1 and 2.3]{BDL90}.
We will show that the map is not compact by finding a sequence of orthogonal vectors
whose norms are all bounded below in $X_0(\a(O)_1)$, i.e. in the image of the unit ball  $\a(O)_1$ by $X_0$.

We now make some general remarks on maps involving relative modular operators
(with notations that will be then suitable to our application).
Let $\M\subset\b(\K)$ be a von Neumann algebra with $\Omega\in\K$ cyclic and separating vector.
From the normal states $\{\phi_q\}$ on $\M$, we take a sequence
converging in norm to $\omega(\cdot)=\langle\Omega,\cdot\,\Omega\rangle$:
conveniently, we shall index this sequence by $\frac{1}{n}$ rather than $n$;
i.e.\! its terms are $\phi_{\frac{1}{n}}$.
As $\M$ is in the standard form, we can find vectors in the natural cone ${\xi_\frac1n}\in\P^\natural(\M,\Omega)$, such that
$\phi_\frac1n(\cdot)=\langle \xi_\frac1n,\cdot\,\xi_\frac1n\rangle$. By convergence of $\phi_\frac1n$, $\xi_\frac1n$ converges in norm to $\Omega$.
We define the maps
\[
X_\frac1n:\M\ni a\mapsto\Delta^{\frac14}_{\M,\xi_\frac1n,\Omega}a\Omega\in\K
\]
\begin{lemma}
$X_\frac1n$ are bounded, *-strongly continuous maps for any $n\in\NN$.
\end{lemma}
\begin{proof}The thesis follows from the fact that for any $a\in\M_1$, we get
\begin{align*} 
\|\Delta^{\frac14}_{\M,\xi_\frac1n,\Omega}a\Omega\|^2 &=\langle \Delta^{\frac14}_{\M,\xi_\frac1n,\Omega}a\Omega,\Delta^{\frac14}_{\M,\xi_\frac1n,\Omega}a\Omega\rangle\\
& = \langle S_{\M,\xi,\Omega} a\Omega, J_{\M,\Omega} a\Omega\rangle \\
&= \langle a^*\xi_\frac1n,J_{\M,\Omega}a\Omega\rangle\\
& \leq \|a^*\xi_\frac1n\|^2+\|a\Omega\|^2,
\end{align*}
where $J_{\M,\Omega}$ is the modular conjugation of $\M$ with respect to $\Omega$.

\end{proof}

Now,  consider the GNS representation of $\MM_2(\CC)$ matrices with respect to the trace state $\Tr$.
$\MM_2(\CC)$ acts on the four dimensional Hilbert space $\CC^4$, where we can fix an orthonormal basis $\{e_{jk}\}_{j,k=1,2}$.
We  define the von Neumann algebra $\widehat \M\dot=\M\otimes \MM_2(\CC)$ acting on  $\widehat\K\dot= \K\otimes \CC^4$.
Consider the vector state $\nu=\Omega\otimes e_{11}+\xi\otimes e_{22}$ where $\Omega$ and $ \xi$
are cyclic and separating vectors for $\M$.  Then $\nu$ is cyclic and separating for $\widehat \M$ and the Tomita operator $S_{\widehat \M,\nu}$ has the following form:
$$S_{\widehat \M,\nu}=U_{11}S_\Omega U_{11}^*+U_{21}S_{\xi,\Omega}U_{12}^*+U_{12}S_{\Omega,\xi}U_{21}^*+U_{22}S_{\xi} U_{22}^*.$$
$S_{\widehat \M,\nu}$ has polar decomposition $S_{\widehat M,\nu}=J_{\widehat \M,\nu}\Delta^{1/2}_{\widehat \M,\nu}$  where
$$\Delta_{\widehat \M,\nu}=U_{11}\Delta_\Omega U_{11}^*+U_{21}\Delta_{\Omega,\xi}U_{21}^*+U_{12}\Delta_{\xi,\Omega}U_{12}^*+U_{22}\Delta_{\xi,\xi} U_{22}^*$$
and
$$J_{\widehat \M,\nu}=U_{11}J_\Omega U_{11}^*+U_{21}J_{\Omega,\xi}U_{12}^*+U_{12}J_{\xi,\Omega}U_{21}^*+U_{22}J_{\xi,\xi} U_{22}^*$$
where $U_{jk}:\K\to\widehat\K$ such that  $U_{jk}\,\eta=\eta\otimes e_{jk}$ with $j,k=1,2$. See \cite{BR1} for further details.

Convergence of linear functionals $\{\phi_\frac1n\}$ implies convergence of the maps $X_\frac1n$ to
$$X_0:\M\ni a\mapsto\Delta^{\frac14}_{\M,\Omega}a\Omega\in\K.$$

\begin{lemma} \label{lem:cont}
$X_\frac1n$ converges in norm to $X_0$ as $n \to \infty$.
\end{lemma}
\begin{proof}
Consider the von Neumann algebras $\widehat \M$ and the vector states on $\M$ implemented by $\widehat\xi_\frac1n\equiv\Omega\otimes e_{11}+\xi_\frac1n\otimes e_{22}$.  As $n$ goes to infinity, $\widehat\xi_\frac1n$ converges in norm to $\widehat\Omega\equiv\Omega\otimes  e_{11}+\Omega\otimes e_{22}$.

 By Lemma 2.7 in \cite{BDL90}, the maps 
$$\widehat X_\frac1n:\widehat\M\ni a\mapsto\Delta^{\frac14}_{\widehat \M,\widehat \xi_\frac1n}a\widehat\xi_\frac1n\in\K,$$
converge in norm to 
$$\widehat X_0:\widehat\M\ni a\mapsto\Delta^{\frac14}_{\widehat\M,\widehat\Omega}a\widehat\Omega\in\K.$$
 Note that the restrictions  $\widehat X_\frac1n|_{\M\otimes e_{1,2}}$ and $\widehat X_0|_{\M\otimes e_{1,2}}$ coincide with $X_\frac1n$ and $X_0$.
 In particular, $\widehat X_\frac1n|_{\M\otimes e_{1,2}}$ converges in norm to $\widehat X_0|_{\M\otimes e_{1,2}}$,
 which is the convergence of $X_\frac1n$ to $X_0$.
\end{proof}

In our case, let $\K = \h\otimes\h, \M=\a(W)=\a(\RR^+)\otimes\a(\RR^-)$ and $O$ be a double cone containing the charge localization, i.e.\! $O\Supset I\times I$. By Lemma \ref{lem:cont}, we learn that 
$$X_\frac1n:\a(O)\ni a\mapsto\Delta^{\frac14}_{\a(W),\xi_\frac1n,\Omega}a\Omega\in\h\otimes\h$$ 
converges in norm to 
$$X_{0}:\a(O)\ni a\mapsto\Delta^{\frac14}_{\a(W),\Omega}a\Omega\in\h\otimes\h.$$
By the Bisognano-Wichmann property, $\Omega$ is the unique eigenvector of $\Delta_{\a(W),\Omega}$ with the eigenvalue $0$,
hence we can find  $a\in\A(O)_1\cap \a(I\times I)'$ such that $\|\Delta^{\frac14}_{\a(W),\Omega}a\Omega\|>0$ and
for such $a$ it holds that 
\[
 \|\Delta^{\frac14}_{\a(W),\xi_\frac1n,\Omega}a\Omega- \Delta^{\frac14}_{\a(W),\Omega}a\Omega\|\stackrel{n\rightarrow 0}\rightarrow 0.
\]
In particular, the sequence has a lower, non zero  norm-bound, i.e.\!
\[
 \liminf_{n\in\NN}\|\Delta^{\frac14}_{\a(W),\xi_\frac1n,\Omega}a\Omega\|>0.
\]
The only non-zero component of the vector $\Delta^{\frac14}_{\tilde\a(W),\tilde\Omega}\tilde\s(a)\psi^{\frac1n}\tilde\Omega$
is the $\h_\frac1n\otimes\h_\frac1n$-component.
Therefore, the following sequence
$\left\{\Delta^{\frac14}_{\tilde\a(W),\tilde\Omega}\tilde\s(a)\psi^{\frac1n}\tilde\Omega\right\}_{n\in\NN}$
with
\[
 \left(\Delta^{\frac14}_{\tilde\a(W),\tilde\Omega}\tilde\s(a)\psi^{\frac1n}\tilde\Omega \right)_\frac1n
 \;=\; \Delta^{\frac14}_{\a(W),\xi_\frac1n,\Omega}a\Omega
 \;\in\; \h_\frac1n\otimes\h_\frac1n \subset \tilde\h
\]
is an orthogonal family of vectors in $X_0(\tilde\a(O)_1)$ with whose norms are uniformly bounded below
by a positive constant (note that we used that $\tilde\s(a) = \bigoplus_{q\in\QQ} a$,
since $a\in\a(O)_1\cap \a(I\times I)'$). Summarizing, we got a sequence of non-convergent
vectors in $X_0(\tilde\a(O)_1)$ and this makes the split property fail.

The work \cite{BD95} gives a sufficient condition for an extension of a split net to be split,
but it does not apply to our situation.

\begin{remark}
 We choose $\QQ$ as the index set because we wanted to have a counterexample on a separable
 Hilbert space.
 The whole construction can be repeated by replacing $\QQ$ by $\RR$,
 and one obtains a two-dimensional conformal net on a non-separable Hilbert space.
 Although we do not attempt to prove it, this case appears to be equivalent to
 the construction of \cite[Section 4]{Ciolli09}, where the same chiral algebra and
 superselection sectors appear.
 
 A two-dimensional Haag-Kastler net on a non-separable Hilbert space cannot satisfy
 the split property: by the Reeh-Schlieder property, an intermediate type I factor must be
 $\s$-finite, while a type I factor is $\s$-finite if and only if it is isomorphic to $\b(\mathcal{K})$
 where $\mathcal{K}$ is separable. If the split property holds, there must be an increasing sequence of
 type I factors which generate the whole $\b(\h)$, which is possible only if
 one of them is isomorphic to $\b(\h)$ (by considering the cardinality), hence $\h$ must be separable.
\end{remark}

\section{Outlook}\label{outlook}

In general, a standard technique to prove the split property
is to verify certain nuclearity conditions for the dynamics. In the M\"obius covariant case, the most handy one
is the trace class condition of the conformal Hamiltonian $e^{-\beta L_0}$ \cite{BDL}.
The split property in turn implies certain compactness conditions \cite{BDL90}.
With our result, one is lead to conjecture that the trace class property should be also automatic.

The existence of an intermediate type I factor does not depend on the sector. Assume $\a$ to be a M\"obius covariant net satisfying split property (for instance $\a$ is a conformal net) and $I_1\subset I_2$ an inclusion of intervals with no common end points. Any representation $\pi$ of $\a$ is a family of local algebra faithful isomorphisms onto their image, as any local algebra is a factor. Then an intermediate type I factor $\a(I_1)\subset \R \subset \a(I_2)$ is mapped through $\rho$ onto an intermediate type I factor $\rho_{I_2}(\a(I_1))\subset\rho_{I_2}(\R)\subset\rho_{I_2}(\a(I_2))$  as $\rho_{I_2}$ restricts to  an isomorphism of $\R$ on $\rho_{I_2}(\R)$. Furthermore, when $\rho$ is localizable, then $\rho_{I_1}(\a(I_1))\subset\rho_{I_2}(\a(I_2))$ is a standard split inclusion acting on a separable Hilbert space (we can unitarily identify the Hilbert spaces). 
At this point it is also natural to expect that the trace class property of $L_0^\rho$ in irreducible or factorial sectors
should be automatic.
While the split property has important implications in algebraic QFT,
it is almost never seen in other approaches to CFT, such as vertex operator algebras (VOAs).
On the other hand, the trace class property, or even the finite-dimensionality of the eigenspaces of $L_0$
would be useful for the study of VOAs.

\subsubsection*{Acknowledgment}
We would like to thank James Tener for calling our attention to the article 
\cite{Neretin} of Neretin, which led us to find \cite{Ols} and allowed us to bridge the gap in the concept of the 
proof we previously had. 
We are grateful to Marcel Bischoff, Sebastiano Carpi, Fabio Ciolli, Luca Giorgetti and Roberto Longo
for various interesting discussions on two-dimensional CFT.

\end{document}